\documentclass[acmsmall]{acmart}

\usepackage[ruled,vlined,linesnumbered]{algorithm2e}
\usepackage{amsthm}
\usepackage{dsfont}
\usepackage{mathtools}
\usepackage{microtype}
\usepackage{multirow}

\newcommand{\modelname}{DAAKG\xspace}
\DeclareMathOperator*{\argmin}{arg\,min}
\DeclareMathOperator*{\argmax}{arg\,max}

\setcopyright{acmlicensed}
\copyrightyear{2023}
\acmYear{2023}
\acmDOI{10.1145/3589304}

\acmJournal{PACMMOD}

\acmYear{2023} \acmVolume{1} \acmNumber{2}
\acmArticle{159} \acmMonth{6}
\acmPrice{15.00}
\acmISBN{978-1-4503-XXXX-X/18/06}

\received{October 2022}
\received[revised]{January 2023}
\received[accepted]{February 2023}

\begin{document}

\title[Deep Active Alignment of Knowledge Graph Entities and Schemata]{Deep Active Alignment of Knowledge Graph Entities and Schemata}

\settopmatter{authorsperrow=3}

\author{Jiacheng Huang}
\affiliation{
    \department{State Key Laboratory for Novel Software Technology} 
    \institution{Nanjing University}
    \country{China}
}
\email{jchuang.nju@gmail.com}

\author{Zequn Sun}
\affiliation{
    \department{State Key Laboratory for Novel Software Technology} 
    \institution{Nanjing University}
    \country{China}
}
\email{zqsun.nju@gmail.com}

\author{Qijin Chen}
\affiliation{
    \institution{Alibaba Group}
    \country{China}
}
\email{qijin.cqj@alibaba-inc.com}

\author{Xiaozhou Xu}
\affiliation{
    \institution{Alibaba Group}
    \country{China}
}
\email{heixia.xxz@alibaba-inc.com}

\author{Weijun Ren}
\affiliation{
    \institution{Alibaba Group}
    \country{China}
}
\email{afei@alibaba-inc.com}

\author{Wei Hu}
\authornote{Wei Hu is the corresponding author.}
\affiliation{
    \department{State Key Laboratory for Novel Software Technology} 
    \department{National Institute of Healthcare Data Science}
    \institution{Nanjing University}
    \country{China}
}
\email{whu@nju.edu.cn} 
\renewcommand{\shortauthors}{Jiacheng Huang et al.}

\begin{abstract}
Knowledge graphs (KGs) store rich facts about the real world. 
In this paper, we study KG alignment, which aims to find alignment between not only entities but also relations and classes in different KGs. 
Alignment at the entity level can cross-fertilize alignment at the schema level.
We propose a new KG alignment approach, called \modelname, based on deep learning and active learning.
With deep learning, it learns the embeddings of entities, relations and classes, and jointly aligns them in a semi-supervised manner.
With active learning, it estimates how likely an entity, relation or class pair can be inferred, and selects the best batch for human labeling.
We design two approximation algorithms for efficient solution to batch selection.
Our experiments on benchmark datasets show the superior accuracy and generalization of \modelname and validate the effectiveness of all its modules.
\end{abstract}




\maketitle

\section{Introduction}
\label{sect:intro}

Knowledge graphs (KGs), which store massive facts about the real world, are used by various parties in many knowledge-driven applications, e.g., semantic search, question answering and recommender systems \cite{KGE,KGSurvey}.
KG alignment is a vital step in knowledge sharing and transfer, which identifies elements (including entities, relations and classes) in different KGs referring to the same real-world thing. 
Recent KG alignment studies~\cite{OpenEA,KGE,EA4KG,BenchmarkEA} mainly focus on entity alignment.
Many methods leverage deep learning techniques to represent entities with low-dimensional embeddings, and align entities with a similarity function on the embedding space trained with seed matches.
Deep alignment methods can resolve the heterogeneity of different KGs. 
They can better compare entities based on the structure information and improve the alignment accuracy.

\begin{figure}
\centering
\includegraphics[width=.8\columnwidth]{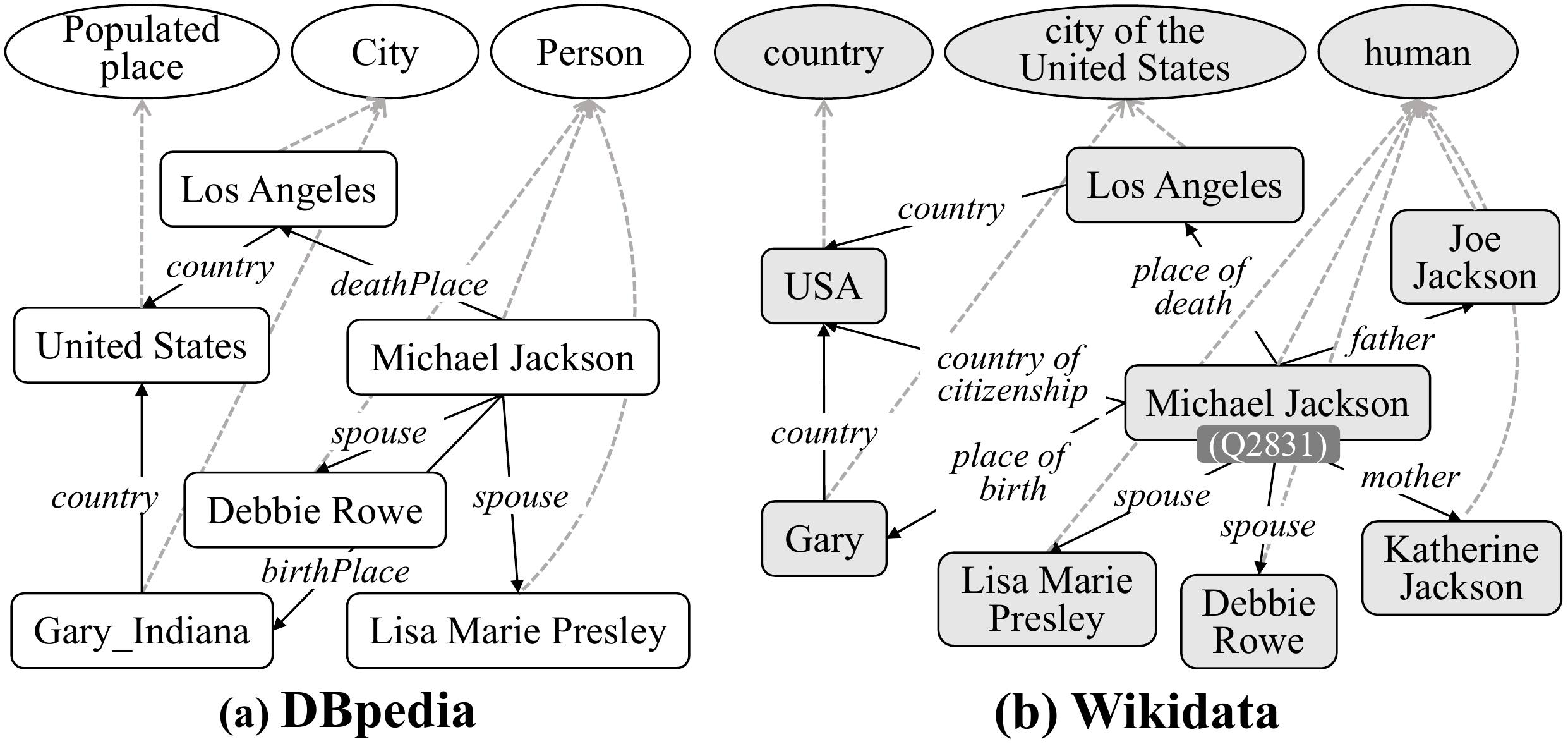}
\caption{An example of KG alignment between DBpedia and Wikidata.
Rectangles denote entities, and ellipses denote classes. 
Solid lines denote relations, and dashed grey lines denote ``\textit{type}''.
\textit{Michael Jackson} in DBpedia and \textit{Q2831} in Wikidata form an entity match.}
\label{fig:example}
\end{figure}

A crucial shortcoming of deep alignment methods is that they demand a large number of seed matches as training data, which may not always be available.
A reasonable solution is active learning, which iteratively asks an oracle or real human annotators to label a small set of training data, and trains a deep alignment model until the labeling budget runs up.
Active (or crowdsourced) entity alignment has been widely studied in the database area, but most existing methods focus on tabular data with literal attributes~\cite{Power,Corleone,Falcon,DIAL,DTAL,DAEM}.
They propose similarity measures or deep learning models to compare literal attributes and generate feature vectors for active learning.
However, entities in KGs differ greatly from those in databases, as different KGs are typically represented by heterogeneous schemata.

\begin{example}\label{ex:kga}
Let us consider a real-world example of KG alignment in Figure~\ref{fig:example}.
Edge $(\textit{Michael}$ $\textit{Jackson},\textit{birthPlace},\textit{Gary\_Indiana})$ in DBpedia \cite{DBpedia} means that the birth place of \textit{Michael Jackson} is \textit{Gary\_Indiana}, while edge $(\textit{Michael Jackson},\textit{place of birth},\textit{Gary})$ in Wikidata \cite{Wikidata} means that the birth place of \textit{Michael Jackson} is \textit{Gary}.
The two entities \textit{Gary\_Indiana} and \textit{Gary} are hard to align based on their literal similarity, e.g., the character-level Jaccard similarity is only $0.44$, but we can confidently infer this pair as a match if we know that they are the birth place of the same person.
\end{example}

This example shows that the active alignment methods for tabular data that only compare entities with literal information cannot be directly applied to KG alignment.
New active alignment methods are required to incorporate edge semantics.
Furthermore, KGs have much more complex schemata than tabular data.
For example, the well-known KG, DBpedia, contains $320$ classes and $1,650$ relations.
Considering such a scale, schema alignment also needs active learning.
Additionally, entity alignment and schema alignment can fertilize each other.
Taking Figure~\ref{fig:example} as an example, relation match $(\textit{birthPlace},\textit{place of death})$ helps an entity alignment model infer entity match $(\textit{Gary\_Indiana}, \textit{Gary})$, and entity matches also help infer class match $(\textit{Person}, \textit{human})$.

In this paper, we propose a full-fledged approach for deep active alignment of KG entities and schemata, called \textbf{\modelname}.
We design a deep alignment model to encode both entity and schema information of KGs for comparing KG elements with their embeddings.
For active learning, we leverage the KG structures and the alignment model to judge which unlabeled element pairs can be inferred with labeled ones, in a similar way as we describe in Example~\ref{ex:kga}.
There are three challenges to fulfilling the deep active alignment.

The first challenge is how to design an embedding-based alignment model.
Although there are a number of existing KG embedding models, they do not carefully capture the schema information~\cite{Universal}.
In fact, the many-to-one problem between entities and classes may degrade the performance of existing models \cite{SIGMOD_KGE,DanglingAware}.
Furthermore, due to the heterogeneity of different KGs, there often exist a bunch of dangling entities that have no matched counterparts \cite{BenchmarkEA}. 
See \textit{Joe Jackson} and \textit{Katherine Jackson} in Figure~\ref{fig:example}.
They would also affect the embedding and alignment learning.
In this paper, we pay attention to encoding the entity-class structure, propose a joint embedding model to align entities, relations and classes simultaneously, and weight entities to reduce the impact of dangling entities.

The second challenge is how to measure the impact of labeled element pairs to unlabeled ones.
Although some relations like \textit{place of birth} can provide strong evidence for entity matches~\cite{PARIS}, most other relations cannot.
We need to assess how well a labeled match can affect the alignment model to infer a neighboring pair.
Our idea is to determine a bound on the embedding difference between them.
When the bound is tight, the neighboring pair is likely to be a match.
In this paper, we propose an inference power measurement to decide which neighboring pair can be safely inferred as a match.

The third challenge is how to efficiently pick a batch of element pairs with the greatest inference power.
As deep models require a large set of labeled pairs as training data, active learning should select a batch of the pairs every time~\cite{ActiveSurvey}.
However, if some pairs can be inferred with others, selecting them to probe the annotators would waste the labeling budget.
We formulate the element pair selection problem in a way that maximizes the expected overall inference power under a given labeling budget.
We prove that it is a sub-modular optimization problem, and give a greedy algorithm to solve it.
We further propose a graph partitioning-based algorithm to improve efficiency.

Our main contributions in this paper are outlined as follows:
\begin{itemize}
    \item We propose a deep KG alignment approach, which is able to jointly align entities, relations and classes. (Sect.~\ref{sect:align})
    
    \item We design an inference power measurement, which gives bounds on the embedding difference to estimate how likely an element pair can be inferred by labeled matches. (Sect.~\ref{sect:utility})
    
    \item We formulate the element pair selection problem for batch active learning, and present an efficient algorithm with an approximation ratio guarantee. (Sect.~\ref{sect:active})
    
    \item We conduct extensive experiments on benchmark datasets to demonstrate the accuracy and generalization of our approach. 
    Source code is accessible online.\footnote{\url{https://github.com/nju-websoft/DAAKG}} (Sect.~\ref{sect:exp}) 
\end{itemize}
\section{Approach Overview}
\label{sect:overview}

In this section, we first introduce the preliminaries for our studied problem. 
Then, we describe the general workflow of our approach. 

\subsection{Preliminaries}
\label{subsect:prelim}

In this paper, we formulate a KG as a quadruple $G=(E,R,C,T)$, where $E,R,C,T$ denote the sets of entities, relations, classes and triplets, respectively.
To simplify the expression, entities, relations and classes are uniformly called \emph{elements}.
A triplet is in the form of $(head, relation, tail)$, where $head$ is an entity, $tail$ is an entity or a class, and $relation$ describes the relation between $head$ and $tail$. 
Oftentimes, a KG defines various classes to describe entity types with a specific relation \textit{type}, such as $(\textit{Michael Jackson}, \textit{type},\textit{singer})$. 
Note that one entity may belong to multiple classes, e.g., \textit{Michael Jackson} can be both a \textit{singer} and a \textit{dancer}. 

\emph{Entity alignment} aims to identify entities in different KGs referring to the same real-world object. 
Let $e$ and $e'$ denote two entities in two different KGs, respectively.
We call $q=(e,e')$ a \emph{match} and denote it by $y^*(q) = y^*(e,e') = 1$ if $e$ and $e'$ refer to an identical object. 
On the contrary, we call $q=(e,e')$ a \emph{non-match} and denote it by $y^*(q)=y^*(e,e')=-1$ if $e$ and $e'$ refer to two different objects. 
Matches and non-matches are both regarded as \emph{labeled} entity pairs, while other pairs are regarded as \emph{unlabeled}.
Similarly, \emph{schema alignment} aims to align relations and classes between different KGs.
In this paper, only equivalent relations and classes are detected.

\emph{Deep alignment} leverages deep learning models to tackle the KG alignment problem \cite{OpenEA,EA4KG}.
It learns element embeddings, i.e., representations in a continuous vector space, and compares elements through their embeddings to judge whether they are a match or a non-match.
The premise is that such embeddings can potentially mitigate data heterogeneity and facilitate knowledge reasoning.
In many cases, a set of labeled element pairs is provided as training data to jump-start the embedding and alignment learning.

\emph{Active alignment} aligns two KGs with human annotators, who iteratively label a few element pairs to help train the alignment model, and aims to achieve the best accuracy under a fixed labeling cost. 
In active alignment, human annotators are widely assumed to be an oracle, which indicates that they can always return the true labels for unlabeled element pairs. 
In this paper, we consider the pool-based active alignment~\cite{DAEM,DIAL,ActiveEA}, which seeks to decide the labels of a fixed unlabeled element pair set $P$, called \emph{pool}.
As our goal is to find all entity and schema matches, given two KGs $G=(E,R,C,T)$ and $G'=(E',R',C',T')$, the full pool can be $P_\text{full}=(E\times E')\cup(R\times R')\cup(C\times C')$.
For large KGs with millions of elements, the full pool may contain up to $10^{12}$ element pairs, which is inefficient or even impossible to execute an active alignment algorithm.
To tackle this issue, previous work~\cite{DAEM,DIAL} often leverages the \emph{blocking} techniques to filter out a small set of element pairs for composing a pool $P\subset P_\text{full}$, and runs active alignment within $P$.

In this paper, we define the studied problem as follows:
\begin{definition}[Deep active KG alignment]
Given two KGs $G=(E,R,C,T)$ and $G'=(E',R',C',T')$, the deep active KG alignment aims to train a deep alignment model $\mathcal{S}: (E\times E')\cup(R\times R')\cup(C\times C') \rightarrow [-1,1]$ by interactively asking a human oracle to label a small set of element pairs, such that $\forall (x,x')\in(E\times E')\cup(R\times R')\cup(C\times C'), \mathcal{S}(x, x')\approx 1 \Leftrightarrow x$ and $x'$ refer to the same thing.
\end{definition}

Table~\ref{tab:notations} summarizes the frequently-used notations in this paper.
Generally, we use normal letters for denoting sets and elements, and bold letters for vectors and matrices.

\begin{table}[!tb]
\centering
\caption{Frequently-used notations}
\label{tab:notations}
\begin{tabular}{cl}
	\hline 	Notations & Descriptions \\
	\hline	$G=(E,R,C,T)$ & KG with entities, relations, classes, triplets \\
	        $e,r,c,t,x$ & entity, relation, class, triplet, element \\
	        $\mathbf{e},\mathbf{r},\mathbf{c}$ & entity embedding, relation embedding, class embedding \\
			$q=(x,x'),y^*(q)$ & element pair, and its true label \\
			$P,L$ & element pair pool, labeled element pairs \\
			$\mathcal{O},\mathcal{S}$ & objective function, deep alignment model \\
			$\mathcal{I},\mathcal{G}$ & inference power, gain \\
	\hline
\end{tabular}
\end{table}

\subsection{Workflow}
\label{subsect:workflow}

Figure~\ref{fig:workflow} illustrates the general workflow of our approach, which accepts two KGs as input and consists of three main modules:
\begin{itemize}

\item\emph{Embedding-based joint alignment.} 
This module contains a KG embedding model and a joint alignment model.
The KG embedding model not only encodes the entity-relation semantics, but also defines an entity-class scoring function to embed the entity-class structures.
The joint alignment model aligns entities, relations and classes based upon the learned embeddings, and leverages potential matches for semi-supervised model training.

\item\emph{Inference power measurement.}
This module constructs an alignment graph to represent the relatedness among element pairs.
For entity pairs, it searches targeting paths in the alignment graph, and estimates the inference power based on the similarity of relation embeddings along the paths and the bounds introduced by KG embedding estimation.
For relation pairs and class pairs, it measures the inference power based on the gradients of joint alignment model.

\item\emph{Batch active learning.} 
This module couples each entity with its $N$-nearest neighbors based on schema signatures, and composes an element pair pool with these entity pairs and all relation/class pairs.
We define the element pair selection problem to maximize the expected overall inference power, and propose a graph partitioning-based algorithm to find a batch of element pairs in the pool for human annotation.
The newly-labeled element pairs are reused to fine-tune the whole embedding-based joint alignment module.
\end{itemize}

The process terminates after the labeling budget runs out, and returns the alignment of entities, relations and classes as output.

\begin{figure}[!t]
\centering
\includegraphics[width=.7\columnwidth]{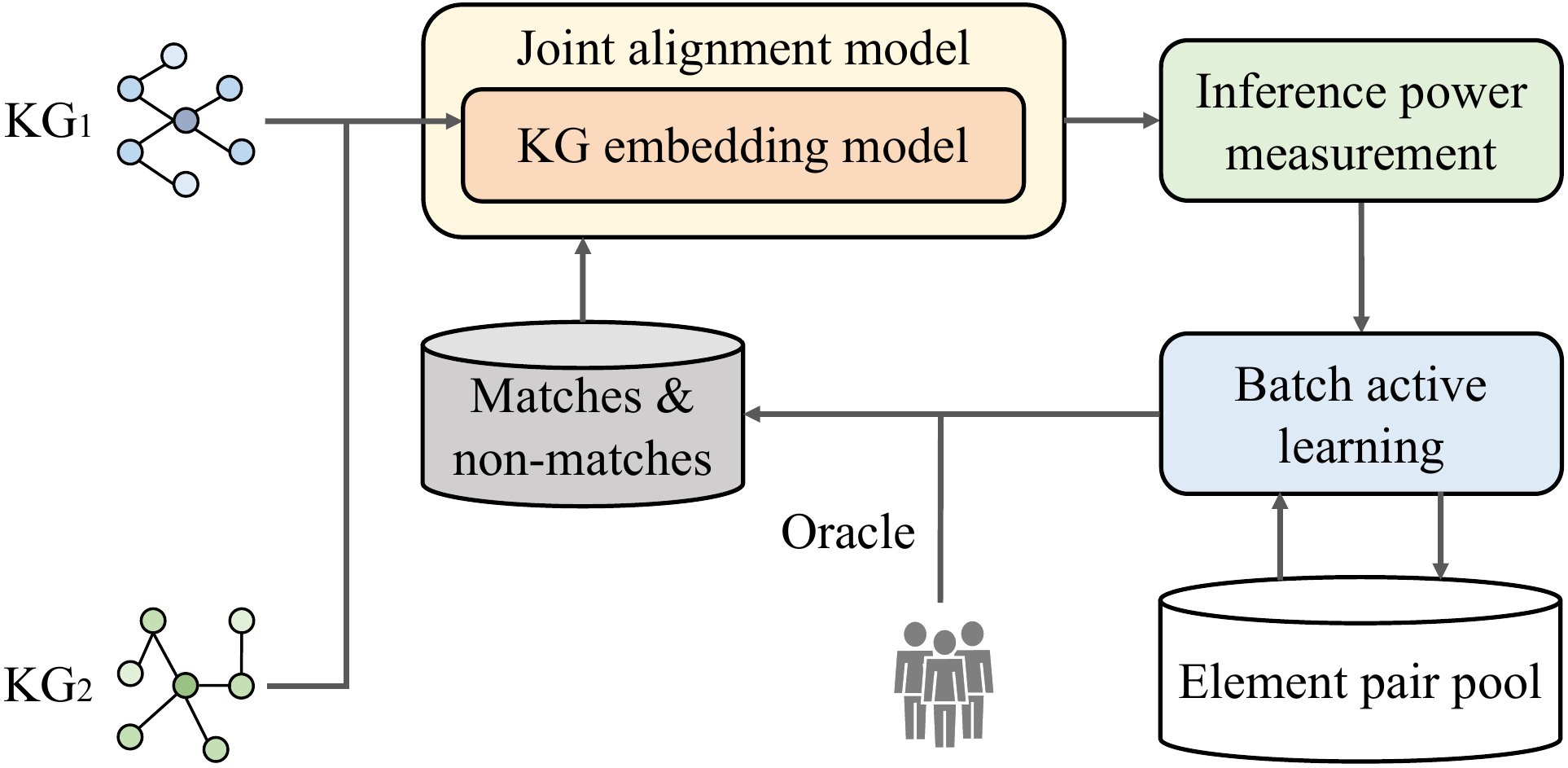}
\caption{Workflow of the proposed approach \modelname}
\label{fig:workflow}
\end{figure}
\section{Related Work}
\label{sect:related}

In this section, we review two lines of related work for KG alignment, namely deep alignment and active alignment.

\subsection{Deep Alignment}
\label{subsect:deep}

In recent years, many deep learning models have been proposed to deal with the entity alignment task.
Compared to conventional approaches \cite{TKDE}, deep learning models represent different KGs in a unified embedding space to overcome heterogeneity and measure entity similarities based on entity embeddings.
As an early attempt, MTransE \cite{MTransE} jointly learns a translation-based embedding model \cite{TKDE} and a transform-based alignment model.
There are typically three improvement lines in later work.
The first line improves the embedding techniques to fit the entity alignment task.
For example, RSN \cite{RSN} extends recurrent neural networks with the skipping mechanism to model long paths across two KGs. 
GCN-Align \cite{GCN-Align}, MuGNN \cite{MuGNN} and KECG \cite{KECG} improve the alignment model by using graph neural networks (GNN) to encode the neighboring information into entity representations.
AliNet \cite{AliNet} and CG-MuAlign \cite{CG-MuAlign} incorporate multi-hop and multi-type neighborhood, respectively, in GNN with the attention and gating mechanisms.
Dual-AMN \cite{Dual-AMN} employs a normalized hard sample mining method to assist the embedding model in distinguishing similar and dissimilar entities.
Note that, since AliNet, CG-MuAlign and Dual-AMN do not learn relation embeddings, they are not applicable to schema alignment.
DAAKG is a transform-based model with GNN. 
Unlike most GNN-based entity alignment methods, DAAKG learns relation and class embeddings based on KG structures for schema alignment.
The second line focuses on effective alignment learning with limited supervision.
For example, BootEA~\cite{BootEA} introduces a bootstrapping method to iteratively mine likely-matched entities as new training data.
Similar to BootEA, DAAKG also leverages semi-supervision to improve the accuracy, and it further uses schema alignment to improve entity alignment.
The third line seeks to retrieve auxiliary or indirect supervision signals from the side information of entities, e.g., attributes \cite{CEA}, textual descriptions \cite{AttrE,MultiKE} and numeric values \cite{IMUSE}, pre-trained language models \cite{BERT-INT}, and visual modalities \cite{EVA,MSNEA}.
Unlike this line of work, we pay attention to KG structures in this paper.
For more information, we refer interested readers to recent surveys and benchmark studies~\cite{OpenEA,KGE,EA4KG,BenchmarkEA}.
Also, some studies~\cite{LargeEA,LIME_vldbj,ClusterEA} use graph partitioning methods to improve the scalability, which is beyond the scope of this paper.

Regarding schema matching, a.k.a. ontology matching, conventional methods \cite{OMBook} conduct sophisticated feature engineering for similarity computation.
Recent work also resorts to deep learning for this task.
DeepAlignment \cite{DeepAlignment} refines pre-trained word vectors,
aiming at deriving the descriptions of types and relations that are tailored to ontology matching.
LogMap-ML \cite{LogMap-ML} extends conventional ontology alignment systems using distant supervision for training ontology embeddings and Siamese neural networks for incorporating richer semantics.
BERTMap \cite{BERTMap} is a more robust system to learn contextual representations based on fine-tuning BERT \cite{BERT} on ontological text for class alignment. 
It can support both unsupervised and semi-supervised settings.
Compared with these schema matching methods, DAAKG relies on KG structures rather than the textual descriptions of classes and relations.

Deep learning also has an increasing impact on database alignment.
DeepER~\cite{DeepER}, DeepMatcher \cite{DeepMatcher}, Seq2SeqMatcher~\cite{Seq2SeqMatcher} and CorDEL~\cite{CorDEL} leverage word embeddings and deep sequence models to compare literal attributes of entities.
Benefitting from the success of pre-trained language models such as BERT~\cite{BERT}, DITTO~\cite{DITTO} and EMTransformer~\cite{EMTransformer} achieve the state-of-the-art performance on aligning tabular data.
The study \cite{AnaBERT} empirically analyzes recent methods and finds that BERT can recognize the structure of entity alignment datasets and extract knowledge from entity descriptions.

\subsection{Active Alignment}

Active alignment for tabular data attracts much attention in database~\cite{Corleone,Falcon,DTAL,DAEM,DIAL,RiskSampling}.
The key component in these methods is the criteria for element pair selection.
Corleone~\cite{Corleone}, Falcon \cite{Falcon} and DTAL~\cite{DTAL} pick the most uncertain entity pairs with the largest information entropy of the model prediction.
DIAL~\cite{DIAL} and DAEM \cite{DAEM} conduct batch active learning by combining uncertainty and diversity.
DIAL employs the BADGE measure~\cite{BADGE}, which clusters example pairs according to gradients and picks the most uncertain ones in each cluster.
DAEM builds a partial order on entity pairs with masked neural networks to measure their impact on other pairs.
Risk~\cite{RiskSampling} estimates the label mis-prediction risk, and lets an oracle examine risky entity pairs.
DIAL, DAEM and Risk can hardly be applied to KG alignment, because they heavily depend on literal features of entities.
Besides, transfer learning~\cite{DTAL} and adversarial learning~\cite{DAEM} can be used to improve model stability when there are only limited labels.

Hike~\cite{Hike} is a hybrid KG alignment method that partitions KG into several entity groups, and converts each group of entities into a table to conduct active learning.
As Hike compares entities with literal attributes, it cannot infer matches based on KG structures.
As far as know, ActiveEA~\cite{ActiveEA} is the only study towards deep active entity alignment of KGs.
It exploits dependencies between entities to measure the uncertainty of each entity as well as its impact on the neighboring entities in the KG, and also leverages a bachelor recognizer to distinguish entities that appear in one KG but not in the other.
Compared to ActiveEA, our approach further considers the diversity of element pairs for batch active learning.
\section{Embedding-based Joint Alignment}
\label{sect:align}

As shown in Figure~\ref{fig:network}, we first introduce a KG embedding model to encode the semantics of entities, relations and classes in a KG. 
Then, we propose a joint alignment model to align entities, relations and classes between two KGs based on the embeddings.

\begin{figure*}
\centering
\includegraphics[width=\textwidth]{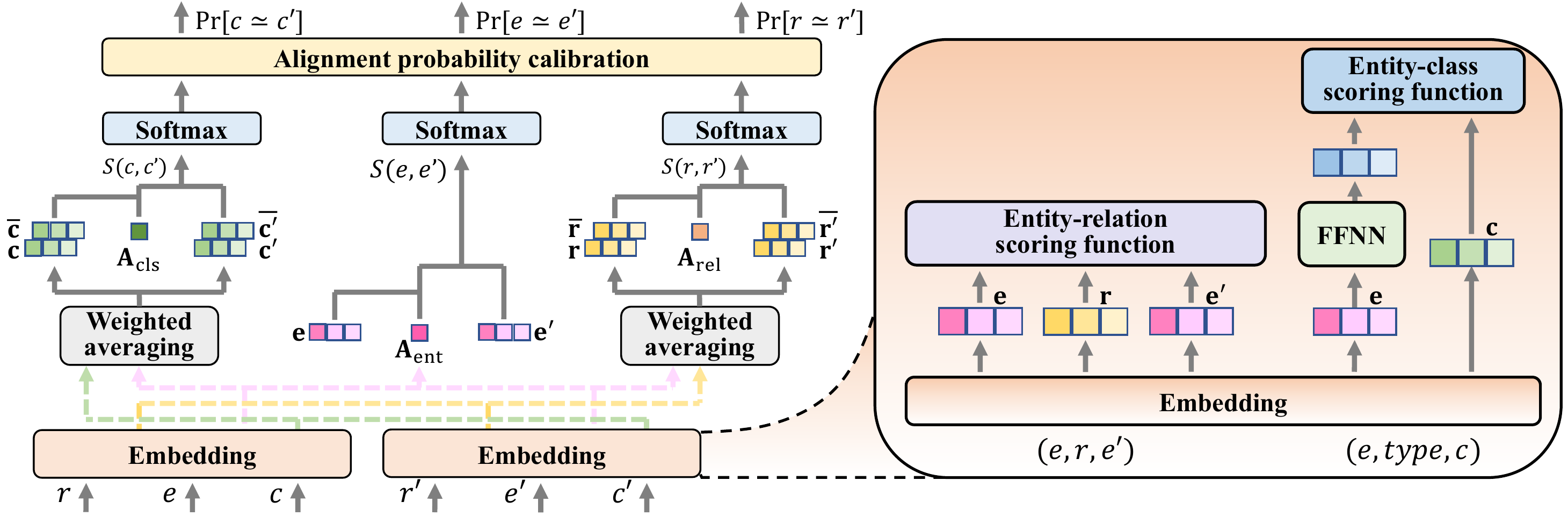}
\caption{Embedding-based joint alignment network}
\label{fig:network}
\end{figure*}

\subsection{Knowledge Graph Embedding Model}

\noindent\textbf{Entity-relation embedding.}
Many KG embedding models~\cite{KGSurvey,KGE} focus on modeling the entity-relation structure of a KG.
They represent entities and relations in a KG $G$ as vectors, and learn a scoring function $f_\text{er}(\mathbf{e},\mathbf{r},\mathbf{e}')$ such that for each triplet $(e,r,e')$ in $G$, we have $f_\text{er}(\mathbf{e},\mathbf{r},\mathbf{e}')\approx 0$; and for each triplet $(e,r,e')$ not in $G$, we have $f_\text{er}(\mathbf{e},\mathbf{r},\mathbf{e}') > 0$.
Since there exist various KG embedding models and new models are constantly being designed, our goal is not limited to a specific one but only requires that the picked model can learn the above scoring function.
Without loss of generality, given a triplet set $T$, we define the loss function for entity-relation embedding, denoted by $\mathcal{O}_\text{er}(T)$, as follows:
\begin{align}
    \mathcal{O}_\text{er}(T) = \sum_{(e,r,e')\in T \vphantom{\tilde{T}}} \smashoperator[r]{\sum_{(e,r,e'')\in \tilde{T}}} \ \big|\, \lambda_\text{er} + f_\text{er}(\mathbf{e},\mathbf{r},\mathbf{e}') - f_\text{er}(\mathbf{e},\mathbf{r},\mathbf{e}'') \,\big|_+,
\end{align}
where $\tilde{T}$ denotes the fake triplet set generated by randomly substituting the tail entity of each triplet in $T$.
It is worth noting that, for each triplet $(e,r,e')$ in $G$, we also add a reverse triplet $(e',r^{-1},e)$, where $r^{-1}$ denotes the synthetic reverse relation of $r$. 
Thus, we only need to replace the tail entities to construct $\tilde{T}$.
$\lambda_\text{er}$ is the margin, and $|\cdot|_+ = \max(\cdot, 0)$.

\smallskip
\noindent\textbf{Entity-class embedding.}
Classes in KGs are used to describe a set of entities, which challenge conventional KG embedding models due to the many-to-one problem \cite{Poincare,TransE}.
In this paper, we model each class as a subspace of the entity embedding space such that for any entity $e$ and class $c$, $e$ belongs to $c$ if and only if the entity embedding $\mathbf{e}$ lies in the subspace determined by $c$. 
Considering the geometric structures of existing entity embedding spaces, in particular the hypersphere \cite{DyERNIE} and hyperbolic spaces~\cite{Hyperbolic,Poincare}, we map each entity embedding into a linear space with a forward-feed neural network (FFNN).
Then, we judge whether the mapped embedding lies in the linear space determined by $c$. 
Specifically, for each entity $e$ and class $c$, we define the scoring function, denoted by $f_\text{ec}(e,c)$, as follows:
\begin{align}
    f_\text{ec}(e,c) = \big|\big|\,\mathbf{W}_c \operatorname{FFNN}(\mathbf{e}) - \mathbf{b}_c\,\big|\big|,
\end{align}
where $\mathbf{W}_c$ and $\mathbf{b}_c$ are learnable parameters for class $c$, and the scoring function yields a subspace $\big\{e\,|\,f_\text{ec}(e,c)\approx 0\big\}$. 
$||\cdot||$ denotes the Euclidean norm of vectors. 
When $e$ belongs to $c$, we expect $f_\text{ec}(e,c)\approx 0$; otherwise, we expect $f_\text{ec}(e,c) > 0$.
This equation enables a large number of entities to be located in the same space for resolving the many-to-one problem. 

For an entity-class triplet set $T_{\text{type}}=\big\{(e,type,c)\in T\big\}$, we define the loss function for entity-class embedding, denoted by $\mathcal{O}_\text{ec}(T_{\text{type}})$, as follows:
\begin{align}
    \mathcal{O}_\text{ec}(T_{\text{type}}) = \smashoperator[l]{\sum_{(e,\textit{type},c)\in T_{\text{type}} \vphantom{\tilde{T}}}} \smashoperator[r]{\sum_{(e',\textit{type},c)\in \tilde{T}_{\text{type}}}}\ \big|\, \lambda_\text{ec} + f_\text{ec}(e,c) - f_\text{ec}(e',c) \,\big|_+,
\end{align}
where $\tilde{T}_\text{type}$ denotes the fake entity-class triplet set generated by randomly substituting entities that do not belong to each class in $T_{\text{type}}$.
$\lambda_\text{ec}$ is the margin.

\subsection{Joint Alignment Model}

\noindent\textbf{Entity alignment.}
The alignment model measures the similarity of two entities based on the cosine similarity of their embeddings.
As entities from different KGs are learned in different embedding spaces, we follow \cite{MTransE,OTEA,DBP2} and leverage a learnable mapping matrix $\mathbf{A}_\text{ent}$ to map the embedding of each entity in a KG $G$ to the embedding space of another KG $G'$. 
Formally, the output of alignment model for an entity pair $(e,e')$ is defined as 
\begin{align}
\mathcal{S}(e, e')=\cos(\mathbf{A}_\text{ent}\,\mathbf{e},\mathbf{e}'), 
\end{align}
where $\mathbf{e},\mathbf{e}'$ denote the embeddings of $e,e'$ in $G$ and $G'$, respectively.

Based on the alignment model, given a labeled entity match set $M_\text{ent}$, we define the loss function for entity alignment, denoted by $\mathcal{O}_\text{ea}(M_\text{ent})$, as follows:
\begin{align}
    \mathcal{O}_\text{ea}(M_\text{ent}) = - \smashoperator[l]{\sum_{(e,e')\in M_\text{ent} \vphantom{\tilde{M}}}} \smashoperator[r]{\sum_{(e'',e''')\in \tilde{M}_\text{ent}}}\ \operatorname{softmax}\big(\mathcal{S}(e,e'),\mathcal{S}(e'',e''')\big),
\end{align}
where $\tilde{M}_\text{ent}$ denotes the entity non-match set obtained by randomly substituting either of two matched entities in $M_\text{ent}$.

\smallskip
\noindent\textbf{Schema alignment.}
The similarities between relations and between classes can be computed based on their embeddings as well. 
However, these embeddings are learned with the embeddings of related entities, which may be affected by dangling entities that have no matched ones in another KG \cite{DBP2}. 
To alleviate the influence of dangling entities, we weight each entity embedding with its similarity to the most similar entity in another KG.
Formally, given two sets of entities $E$ and $E'$ in different KGs, the weight for an entity $e\in E$, denoted by $w_e$, is defined as
\begin{align}
    w_e = \max_{e'\in E'}\mathcal{S}(e, e').
\end{align}

Symmetrically, the weight for each entity $e' \in E'$ is defined as $w_{e'} = \max_{e\in E}\mathcal{S}(e, e')$.

For a relation $r$, each of its related triplets determines a local optimum relation embedding. 
We softly average all the local optimum embeddings of relation $r$ as its feature to reduce the impact of dangling entities.
Specifically, given a relation $r$, for each triplet $(e, r, e')\in T$, we find a local optimum relation embedding by minimizing $f_\text{er}(\mathbf{e},\mathbf{r},\mathbf{e}')$, and compute the mean embedding from all local optimums, denoted by $\bar{\mathbf{r}}$, as follows:
\begin{align}
    \bar{\mathbf{r}} = \frac{\sum_{(e,r,e')\in T}\min(w_e, w_{e'}) \argmin_{\mathbf{r}}f_\text{er}(\mathbf{e},\mathbf{r},\mathbf{e}')}{\sum_{(e,r,e')\in T}\min(w_e, w_{e'})}.
\end{align}

Note that $\min(w_e, w_{e'})$ implies that, when $e$ or $e'$ is a dangling entity, the triplet $(e,r,e')$ would be softly removed.

When calculating the similarity of two relations $r$ and $r'$, we pick the larger cosine similarity from relation embeddings and mean embeddings, which is $\mathcal{S}(r,r')=\max\big(\cos(\mathbf{A}_\text{rel}\,\mathbf{r},\mathbf{r'}),\cos(\mathbf{A}_\text{ent}\,\bar{\mathbf{r}},\bar{\mathbf{r}'})\big)$, where $\mathbf{A}_\text{rel}$ is a learnable mapping matrix.
Note that, as the mean embeddings are obtained from entity embeddings, we use $\mathbf{A}_\text{ent}$ to map them.
The loss function for relation alignment is 
\begin{align}
\mathcal{O}_\text{ra}(M_\text{rel}) = - \smashoperator[l]{\sum_{(r,r')\in M_\text{rel} \vphantom{\tilde{M}}}} \smashoperator[r]{\sum_{(r'',r''')\in \tilde{M}_\text{rel}}}\ \operatorname{softmax}\big(\mathcal{S}(r,r'),\mathcal{S}(r'',r''')\big), 
\end{align}
where $M_\text{rel}$ is the labeled relation match set, and $\tilde{M}_\text{rel}$ is the relation non-match set generated by randomly substituting either of two matched relations in $M_\text{rel}$.

For a class $c$, we softly average the embeddings of all the entities belonging to $c$. 
Therefore, the mean embedding of $c$, denoted by $\bar{\mathbf{c}}$, is defined as 
\begin{align}
    \bar{\mathbf{c}} = \frac{\sum_{(e,\textit{type},c)\in T_{\text{type}}} w_e\,\mathbf{e}}{\sum_{(e,\textit{type},c)\in T_{\text{type}}} w_e}.
\end{align}

When computing the similarity of two classes $c$ and $c'$, we pick the larger cosine similarity from class embeddings and mean embeddings, i.e., $\mathcal{S}(c,c')=\max\big(\cos(\mathbf{A}_\text{cls}\,\mathbf{c},\mathbf{c'}),\cos(\mathbf{A}_\text{ent}\,\bar{\mathbf{c}},\bar{\mathbf{c}'})\big)$, where $\mathbf{A}_\text{cls}$ is a learnable mapping matrix.
Again, we use $\mathbf{A}_\text{ent}$ to map the mean embeddings.
The loss function $\mathcal{O}_\text{ca}(M_\text{cls})$ can be defined similarly, where $M_\text{cls}$ is the labeled class match set.

\smallskip
\noindent\textbf{Semi-supervised model training and fine-tuning.}
Previous work \cite{BootEA,OpenEA} shows that leveraging potential entity matches can improve the performance of entity alignment.
We believe that similar conclusions hold for potential relation and class matches.
Specifically, we pick the element pairs with similarity scores that exceed a threshold $\tau$, and leverage them as additional supervision signals.
To deal with the conflicts in element pairs, e.g., one entity is matched with multiple others, we discard the pairs with lower similarity scores, and denote the remaining ones by $M_{\text{semi}}$.
Moreover, some element pairs in $M_{\text{semi}}$ can be non-matches. 
For each element pair $(x,x')$, we use the alignment model learned in the previous step, denoted by $\mathcal{S}_0(\cdot)$, to compute its similarity as soft label, and define the semi-supervised loss function, $\mathcal{O}_\text{semi}(M_{\text{semi}})$, as follows: 
\begin{align}
\mathcal{O}_\text{semi}(M_{\text{semi}}) = - \sum_{(x,x')\in M_\text{semi}} \mathcal{S}_0(x,x') \mathcal{S}(x,x').
\end{align}

Note that $\mathcal{S}_0(x,x')$ is treated as a constant in the model training, so the model optimizer does not update the parameters in $\mathcal{S}_0$.

Active learning iteratively adds newly-labeled element pairs in the training set.
However, re-training the whole joint alignment model from scratch is wasteful.
In this paper, we propose a fine-tuning process that focuses on adjusting the parameters of newly-labeled element pairs, instead of treating all element pairs equally.
Intuitively, the newly-labeled element pairs are more likely to be misclassified than the existing ones by the trained joint alignment model.
We use the focal loss~\cite{Focal}, a modified cross-entropy loss, to change the softmax output $\operatorname{softmax}\big(\mathcal{S}(x,x'),\mathcal{S}(x'',x''')\big)$ in $\mathcal{O}_\text{ea}$, $\mathcal{O}_\text{ra}$ and $\mathcal{O}_\text{ca}$ by $\big(1-\Pr[x\simeq x']\big)^\gamma$, where $\gamma > 0$ is the focus parameter, so that the loss function would focus more on misclassified newly-labeled element pairs. 
We set $\gamma = 2$ as suggested in \cite{Focal}.

\smallskip
\noindent\textbf{Alignment probability calibration.}
Modern neural networks may produce less-calibrated probabilities on binary classification tasks \cite{Calibration}. 
To accurately estimate whether two elements are a match, we use temperature scaling to convert embedding similarities into alignment probabilities.
For simplicity, we describe the process of entity alignment, and apply the same process to schema alignment.

We formulate entity alignment as a classification problem that classifies each entity $e$ in one KG as a matched entity $e'$ in another KG.
In this way, we can obtain the alignment probabilities as the softmax output of classification results with temperature scaling.
Formally, given two entity sets $E$ and $E'$ in different KGs, the probability that $e\in E$ is classified as a candidate entity $e'\in E'$, denoted by $\Pr[e'\,|\,e]$, can be defined as
\begin{align}
\label{eq:mp}
    \Pr[e'\,|\,e] = \frac{ \exp\big(\frac{\mathcal{S}(e, e')}{Z_\text{ent}} \big)}{ \sum_{e''\in E'} \exp\big(\frac{\mathcal{S}(e, e'')}{Z_\text{ent}}\big) },
\end{align}
where $Z_\text{ent}$ is the temperature parameter, which is used to control the output distribution such that, when the temperature is lower, the output distribution is more discriminatory; when the temperature is higher, the output distribution is less discriminatory.
As there may exist a very large number of entities in a KG, a small $Z_\text{ent}$ is assigned in Eq.~(\ref{eq:mp}).

Furthermore, for each entity pair $(e, e')$, we consider both alignment directions, i.e., classifying $e$ as an entity in $E'$ and classifying $e'$ as an entity in $E$. 
Therefore, we define the alignment probability of $(e, e')$, denoted by $\Pr[y^*(e,e')=1]$, as
\begin{align}
\label{eq:prob}
    \Pr[y^*(e,e')=1]=\min\big\{\Pr[e'\,|\, e], \Pr[e\,|\,e']\big\}.
\end{align}

The following active learning uses $\Pr[y^*(e,e')=1]$ as the weight of $(e, e')$ and expects matches having large weights.
When $\Pr[e'\,|\, e]$ or $\Pr[e\,|\,e']$ is small, $(e, e')$ is likely to be a non-match.
To avoid non-matches in active learning, we choose the smaller probability from $\Pr[e'\,|\, e]$ and $\Pr[e\,|\,e']$ to achieve a more conservative classification estimate.
The alignment probability of each relation pair and that of each class pair are defined similarly.

\smallskip
\noindent\textbf{Parameter complexity.}
The dimensions of entity, relation and class embeddings are denoted by $d_e$, $d_r$, $d_c$, respectively.
The entity-class embedding model contains $|C\cup C|d_c$ parameters to represent classes, and $d_e d_c$ parameters to map entity embeddings.
The alignment model for entities, relations and classes requires $d_e^2$, $d_r^2$ and $d_c^2$ parameters, respectively, for the mapping matrices.
Note that the size of mapping matrices is fixed even though the number of elements increases.
The overall parameter complexity also depends on the used KG embedding model. 
For example, TransE~\cite{TransE} totally uses $|E\cup E'|d_e+|R\cup R'|d_r+|C\cup C'|d_c+d_e d_c+d_e^2+d_r^2+d_c^2$ parameters. 
\section{Inference Power Measurement}
\label{sect:utility}

In this section, we first construct an alignment graph to model the relatedness among element pairs. 
Then, we leverage the alignment graph to measure the inference power of element pairs.

\subsection{Alignment Graph}

Graph structures are widely used to model the alignment states of element pairs and the relatedness among them \cite{REMP,Progressive}.
In this paper, we build an alignment graph by linking the element pairs in the pool if and only if the corresponding elements are connected in the respective KG.
Given two KGs $G=(E,R,C,T), G'=(E',R',C',T')$ and the element pair pool $P$, the alignment graph, denoted by $G\times_P G'$, is a quadruple defined as $\big((E\times E')\cap P, (R\times R')\cap P,(C\times C')$ $\cap\, P, T''\big)$, where $T''=\Big\{\big((x,x'), (r,r'), (x'', x''')\big)\,\big|\,(x,r,x'')\in T\wedge (x',r',x''')\in T'\wedge (x,x')\in P\wedge (r,r')\in P\wedge (x'',x''')\in P \,\Big\}$ denotes the triplet set (including \textit{type}) restricted by $P$, and $x,x',x'',x'''$ are entities or classes.
Figure~\ref{fig:align_graph} illustrates the alignment graph built from Figure~\ref{fig:example}.

\begin{figure}
\centering
\includegraphics[width=.8\columnwidth]{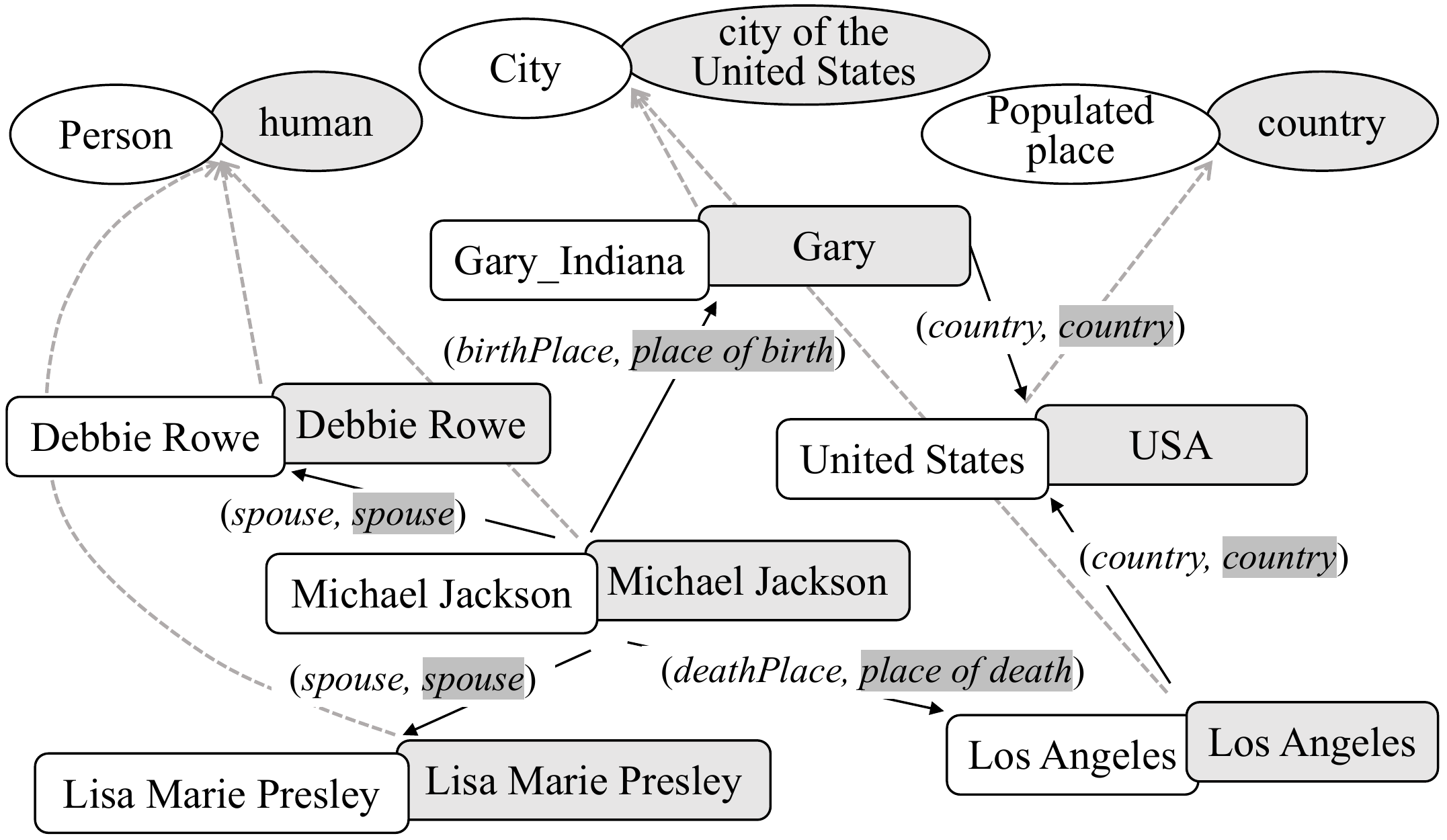}
\caption{The alignment graph built based on Figure~\ref{fig:example}}
\label{fig:align_graph}
\end{figure}

\subsection{Structure-based Inference Power}
The actively-trained joint alignment model discovers new matches based on newly-labeled pairs. 
In order to decide which pairs can be inferred without re-training the model, we measure the structure-based inference power, which indicates how the entity embeddings in labeled pairs affect those in unlabeled pairs.

\smallskip
\noindent\textbf{Inference to entity pairs.}
In the alignment graph, entity pairs are connected through relation pairs. 
The structure-based inference aims to estimate the impact of new labels to the unlabeled entity pairs in terms of the alignment graph.
Suppose that $(e_1,e'_1)$ is a newly-labeled entity match, we consider an edge in the alignment graph, denoted by $(e_1,e'_1) \xrightarrow{(r_1,r'_1)} (e_2,e'_2)$, and infer the label of $(e_2, e'_2)$.
Intuitively, if $r_1$ and $r'_1$ are similar enough, and there are no alternative entities for $e_2$ or $e'_2$ in the KGs, we can judge $(e_2,e'_2)$ as a match.
In order to formulate these two conditions, we leverage the related score functions $f_\text{er}(\mathbf{e}_1,\mathbf{r}_1,\mathbf{e}_2)$ and $f_\text{er}(\mathbf{e}'_1,\mathbf{r}'_1,\mathbf{e}'_2)$.
Note that the embedding-based joint alignment module minimizes the loss function such that $f_\text{er}(\mathbf{e}_1,\mathbf{r}_1,\mathbf{e}_2)\approx 0$ and $f_\text{er}(\mathbf{e}'_1,\mathbf{r}_1',\mathbf{e}'_2)\approx 0$.
We fix the embeddings of entities $e_1$ and $e'_1$, and solve $f_\text{er}(\mathbf{e}_1,\mathbf{r}_1,\mathbf{e}_2)=0, f_\text{er}(\mathbf{e}'_1,\mathbf{r}'_1,\mathbf{e}'_2)=0$ to approximate the embeddings of entities $e_2$ and $e'_2$.
First, we process $f_\text{er}(\mathbf{e}_1,\mathbf{r}_1,\mathbf{e}_2)=0$.
Considering that it is intractable to obtain the exact solution when $f_\text{er}(\cdot)$ is a deep model, and there may exist multiple solutions to $\mathbf{e}_2$, we alternatively find a vector $\mathbf{\tilde{e}}_2$ to approximate $\mathbf{e}_2$ within an error bound $d_1$, i.e., $||\,\mathbf{e}_2-\mathbf{\tilde{e}}_2\,||\leq d_1$.
Inspired of TransE, we introduce a difference vector, denoted by $\mathbf{\tilde{r}}_1 = \mathbf{\tilde{e}}_2 - \mathbf{e}_1$, and we have 
\begin{align}
    ||\,\mathbf{e}_2-(\mathbf{e}_1 + \mathbf{\tilde{r}}_1)\,||\leq d_1.
\end{align}

Similarly, we can find a vector $\mathbf{\tilde{r}}'_1$ and a real number $d'_1$ such that $ \big|\big|\,\mathbf{e}'_2-(\mathbf{e}'_1 + \mathbf{\tilde{r}}'_1)\,\big|\big|\leq d'_1$.

When $f_\text{er}(\cdot)$ is a simple geometric model such as TransE~\cite{TransE}, we can simply obtain that $\mathbf{\tilde{r}}_1=\mathbf{r}_1$ and $d_1=0$.
However, when $f_\text{er}(\cdot)$ is a sophisticated deep neural model such as CompGCN~\cite{CompGCN}, we estimate $\mathbf{\tilde{r}}_1$ and $d_1$ by finding $m$ approximate solutions, computing their mean embedding as $\mathbf{\tilde{r}}_1$, and getting the largest distance from one solution to $\mathbf{e}_1 + \mathbf{\tilde{r}}_1$ as $d_1$.
More specifically, we randomly sample $m$ entity embeddings as the initial values of $\mathbf{e}_2$, use the stochastic gradient descent algorithm to get $m$ local optimum solutions $\mathbf{e}_{2,i}$ for $\min_{\mathbf{e}_2} f_\text{er}(\mathbf{e}_1,\mathbf{r}_1,\mathbf{e}_2)$, and $\mathbf{\tilde{r}}_1$ and $d_1$ are defined as
\begin{align}
    \mathbf{\tilde{r}}_1 = \mathbf{\tilde{e}}_2 - \mathbf{e}_1, \quad
    d_1 = \max_{1\leq i\leq m}\big(||\,\mathbf{e}_{2,i}-\mathbf{\tilde{e}}_2\,||\big), 
\end{align}
where $\mathbf{\tilde{e}}_2 = \frac{1}{m} \sum_{i=1}^m \mathbf{e}_{2,i}$.

Based on the above inequations, we can bound the difference between the embedding of $e'_2$ and the mapped embedding of $e_2$: 
\begin{align}
\label{eq:onehop}
\big|\big|\,\mathbf{A}_\text{ent}\,\mathbf{e}_2-\mathbf{e}'_2\,\big|\big| \leq \big|\big|\,\mathbf{A}_\text{ent}\,\mathbf{e}_1-\mathbf{e}'_1\,\big|\big| +  \big|\big|\,\mathbf{A}_\text{rel}\,\mathbf{\tilde{r}}_1-\mathbf{\tilde{r}}'_1\,\big|\big| + d_1 + d'_1.
\end{align}

Because $(e_1,e'_1)$ is a match, we have $\mathbf{A}_\text{ent}\,\mathbf{e}_1\approx \mathbf{e}'_1$, i.e., $||\,\mathbf{A}_\text{ent}\,\mathbf{e}_1 - \mathbf{e}'_1 \,|| \approx 0$.
Hence, the right term of Eq.~(\ref{eq:onehop}) can be reformulated to $||\,\mathbf{A}_\text{rel}\,\mathbf{\tilde{r}}_1-\mathbf{\tilde{r}}'_1\,|| + d_1 + d'_1$, where $||\,\mathbf{A}_\text{rel}\,\mathbf{\tilde{r}}_1-\mathbf{\tilde{r}}'_1\,||$ denotes the difference of relations $r_1$ and $r'_1$, and $d_1 + d'_1$ denotes the size of the space of possible entities.
When $||\,\mathbf{A}_\text{rel}\,\mathbf{\tilde{r}}_1-\mathbf{\tilde{r}}'_1\,|| + d_1 + d'_1$ is small enough, we know that $||\,\mathbf{A}_\text{ent}\,\mathbf{e}_2-\mathbf{e}'_2\,||$ is small enough, which indicates that $(e_2,e'_2)$ is also a match.
Note that this condition indicates that (1) $||\,\mathbf{A}_\text{rel}\,\mathbf{\tilde{r}}_1-\mathbf{\tilde{r}}'_1\,||$ is small, i.e., $r_1$ and $r'_1$ are similar, and (2) $d_1 + d'_1$ is small, i.e., there are no alternative entities for $e_2$ or $e'_2$.
These two conditions coincide with the intuition mentioned above.

To further extend, when two entity pairs $(e_1,e'_1)$ and $(e_K,e'_K)$ are connected via a path $(e_1,e'_1)$ $\xrightarrow{(r_1,r'_1)}(e_2,e'_2)\xrightarrow{(r_2,r'_2)}(e_3,e'_3)\dashrightarrow (e_{K-1},e'_{K-1})\xrightarrow{(r_{K-1},r'_{K-1})}(e_K,e'_K)$, we can still decide if $(e_K,e'_K)$ is a match given $(e_1,e'_1)$ being a match.
Following the similar approximation, we can obtain $K$ bounds such that $||\,\mathbf{e}_{k+1}-(\mathbf{e}_{k} + \mathbf{\tilde{r}}_k)\,||\leq d_k$, where $k\in\{1,2,\dots,K-1\}$.
By combining these bounds together, we have
\begin{align}
    \Big|\Big|\,\mathbf{e}_K-(\mathbf{e}_1 + \sum_{k=1}^{K-1} \mathbf{\tilde{r}}_k)\,\Big|\Big| \leq \sum_{k=1}^{K-1} d_k.
\end{align}

Symmetrically, we have $||\,\mathbf{e}'_K-(\mathbf{e}'_1 + \sum_{k=1}^{K-1} \mathbf{\tilde{r}}'_k)\,|| \leq \sum_{k=1}^{K-1} d'_k$.
Then, we can define the total difference on the path $(e_1,e'_1)\dashrightarrow(e_K,e'_K)$, denoted by $\mathcal{D}\big((e_1,e'_1)\dashrightarrow(e_K,e'_K)\big)$, as follows:
\begin{align}
    \mathcal{D}\big((e_1,e'_1)\dashrightarrow(e_K,e'_K)\big)  = \Big|\Big|\,\mathbf{A}_\text{rel}\sum_{k=1}^{K-1} \mathbf{\tilde{r}}_{k}-\sum_{k=1}^{K-1} \mathbf{\tilde{r'}}_{k}\,\Big|\Big| + \sum_{k=1}^{K-1} (d_k + d'_k).
\end{align}

Therefore, we can bound the difference between $e_K$ and $e'_K$ by 
\begin{align}
    \big|\big|\,\mathbf{A}_\text{ent}\,\mathbf{e}_K-\mathbf{e}'_K\,\big|\big| \leq \big|\big|\,\mathbf{A}_\text{ent}\,\mathbf{e}_1-\mathbf{e}'_1\,\big|\big| + \mathcal{D}\big((e_1,e'_1)\dashrightarrow(e_K,e'_K)\big).
\end{align}

Again, when $\mathcal{D}\big((e_1,e'_1)\dashrightarrow(e_K,e'_K)\big)$ is sufficiently small, we can judge $(e_K,e'_K)$ as a match.
Note that there may exist several paths from $(e_1,e'_1)$ to $(e_K,e'_K)$. 
The path with the smallest $\mathcal{D}()$ provides the tightest bound.
Thus, the inference power from $(e_1,e'_1)$ to $(e_K,e'_K)$, denote by $\mathcal{I}\big((e_K,e'_K)\,|\,(e_1,e'_1)\big)$, is defined as 
\begin{align}
    \mathcal{I}\big((e_K,e'_K)\,|\,(e_1,e'_1)\big) = \quad \smashoperator[l]{\max_{(e_1,e'_1)\dashrightarrow (e_K,e'_K)}} \Big(\frac{1}{1+\mathcal{D}\big((e_1,e'_1)\dashrightarrow (e_K,e'_K)\big)}\Big).
\end{align}

When $\mathcal{I}\big((e_K,e'_K)\,|\,(e_1,e'_1)\big)$ is large enough, there is a path providing the bound such that when $(e_1,e'_1)$ is a match, $(e_K,e'_K)$ is also a match.
To save the running time, we only consider the paths within $\mu$-hops in this paper.

Another case is that, when a relation pair is labeled as a match, we can set $\big|\big|\,\mathbf{A}_\text{rel}\,\mathbf{\tilde{r}}_1 - \mathbf{\tilde{r}}'_1 \,\big|\big| = 0$ in Eq.~(\ref{eq:onehop}).
As there may be multiple source entity matches reaching $(e_2,e'_2)$ with relation $(r_1,r'_1)$, e.g., $(e_2,e'_2)=(\textit{United States},\textit{USA})$ and $(r_1,r'_1)=(\textit{country},\textit{country})$ in Figure~\ref{fig:align_graph}, we pick the best one, and the inference power from a relation pair to an entity pair can be defined as
\begin{align}
    \mathcal{I}\big((e_2,e'_2)\,|\,(r_1,r'_1)\big) = \quad \smashoperator[l]{\max_{(e_1,e'_1)\text{ is a match}}}\Big(\frac{1}{1+\mathcal{D}\big((e_1,e'_1)\xrightarrow{(r_1,r'_1)} (e_2,e'_2)\big)}\Big).
\end{align}

When $\mathcal{I}\big((e_2,e'_2)\,|\,(r_1,r'_1)\big)$ is large enough, there is an edge from an entity match providing the bound such that $(e_2,e'_2)$ is a match.

\smallskip
\noindent\textbf{Inference to class pairs and relation pairs.}
Given two classes $c$ and $c'$, their similarity $\mathcal{S}(c, c')$ is affected by the belonging entity pairs and their weights.
The inference power of an entity pair $(e,e')$ to $(c,c')$ is defined as its impact on $\mathcal{S}(c, c')$.
The gradient can indicate the impact of entity embeddings to class alignment, and we compute it based on the alignment model, which is
\begin{align}
    \mathcal{I}\big((c,c')\,|\,(e,e')\big) = \big|\big|\,\nabla{\mathbf{e},\mathbf{e}'} \mathcal{S}(c,c')\,\big|\big|.
\end{align}

Given two relations $r$ and $r'$, it requires to know two pairs of entities linked with $(r, r')$ in the alignment graph.
When there is an edge $(e,e')\xrightarrow{(r,r')}(e'',e''')$ and $(e,e')$ is a match, the inference power to $(r,r')$ can also be computed based on the gradient:
\begin{align}
    \mathcal{I}\big((r,r')\,|\,(e,e')\big) = \big|\big|\,\nabla_{\mathbf{e}''-\mathbf{e}, \mathbf{e}'''-\mathbf{e'}} \mathcal{S}(r,r')\,\big|\big|.
\end{align}

\smallskip
\noindent\textbf{Overall inference power.}
Assume that there exist a set of labeled element matches $L^+$.
As different matches in $L^+$ possess different inference power to an unlabeled element pair $q'$, we compute the inference power to $q'$ as the greatest inference power from one element pair $q\in L^+$ to $q'$, i.e., $\mathcal{I}(q'\,|\,L^+) = \max_{q\in L^+} \mathcal{I}(q'\,|\,q)$.

Furthermore, we define the overall inference power of $L^+$ to the element pair pool $P$, denoted by $\mathcal{I}(P\,|\,L^+)$, as follows: 
\begin{align}\label{eq:over_inf}
    \mathcal{I}(P\,|\,L^+) = \sum_{q'\in P :  \mathcal{I}(q'\,|\,L^+) > \kappa}\  \mathcal{I}(q'\,|\,L^+),
\end{align}
where $\kappa$ denotes the inference power threshold.
Again, $q$ (and $q'$) can be an entity, relation or class pair.
\section{Batch Active Learning}
\label{sect:active}

As KGs may have millions of elements (especially entities),
we first generate a candidate pool to reduce the search space of unlabeled element pairs. 
Then, we formulate the element pair selection problem, and propose two approximation algorithms to pick the best element pairs for human annotation.

\subsection{Element Pair Pool Generation}

Classes in KGs naturally partition entities into different groups.
For example, an entity of class \textit{Person} cannot match another entity of class \textit{Location}.
Similarly, relations can also be used to partition entities \cite{JAPE}.
Based on this characteristic of classes and relations, conventional entity blocking methods~\cite{PBA,Hike} use schema alignment to partition entities with matched classes or relations into the same group, and construct a candidate pool with all pairs of entities in each group.
However, it is difficult for these methods to handle dangling classes or relations that do not have a matched counterpart in the other KG.

Benefiting from KG embeddings, we softly compare the schema information of entities.
Specifically, we generate a \emph{schema signature} for each entity, and calculate the cosine similarity of signatures to find entity pairs with similar schema information.
Given two KGs $G=(E,R,C,T)$ and $G'=(E',R',C',T')$, the schema signature of an entity $e\in E$, denoted by $\operatorname{sig}(e)$, is defined as
\begin{align}
\label{eq:sig}
\operatorname{sig}(e) & = \bigg[
\frac{\sum_{r\in R\,:\,(e,r,e'')\in T} w_r\,\bar{\mathbf{r}}}{\sum_{r\in R\,:\,(e,r,e'')\in T} w_r}; \frac{\sum_{c\in C\,:\,(e,type,c)\in T} w_c\,\bar{\mathbf{c}}}{\sum_{c\in C\,:\,(e,type,c)\in T} w_c}\bigg],\\
w_r & = \max_{r'\in R'}\mathcal{S}(r,r'),\quad w_c = \max_{c'\in C'}\mathcal{S}(c,c'),
\end{align}
where $w_r$ is the weight of relation $r$ and $w_c$ is the weight of class $c$, which are used to reduce the influence of dangling relations and classes, respectively, because dangling ones are likely to have low similarities to others.
The schema signatures of entities in $E'$ can be generated symmetrically.

We obtain the element pair pool $P$ by picking the top-$N$ nearest neighbors for each entity and preserving all relation and class pairs, $P=(R\times R')\bigcup\,(C\times C')\bigcup\bigg(\Big(\cup_{e\in E}\big(\{e\}\times \text{top-}N(e)\big)\Big) \bigcap \Big(\cup_{e'\in E'}\big(\text{top-}N(e')\times\{e'\}\big)\Big)\bigg)$, where top-$N$ is ranked based on the cosine similarity of schema signatures.

\subsection{Element Pair Selection Algorithms}

Training deep KG alignment models often requires a large amount of labels. 
A single labeled element pair can hardly improve the performance of deep models.
Furthermore, unlabeled elements are annotated by humans in real-world settings. 
Assigning several unlabeled element pairs to multiple humans simultaneously can reduce the latency.
Therefore, we consider the selection of a batch of unlabeled element pairs every time to annotate.

Let $Q$ be a set of selected element pairs to be annotated by an oracle, and the true label of an element pair $q\in Q$ is denoted by $y^*(q)$, which is unknown.
After $Q$ is labeled, we use the labeled element matches $Q^+=\{q\in Q\,|\,y^*(q)=1\}$ to infer the labels of the rest of element pairs in pool $P$.
Because we cannot truly determine $Q^+$ unless the oracle labels it, we evaluate the \emph{expected} inference power from $Q$ to element pair $q'\in P$: $\mathds{E}_{\mathds{P}(Q^+\,|\,Q)}\big[\mathcal{I}(q'\,|\,Q^+)\big]$, where $\mathds{P}(Q^+\,|\,Q)$ is the distribution of $Q^+$ given $Q$.
To obtain $\mathds{P}(Q^+\,|\,Q)$, we first compute $\Pr[y^*(q)=1]$ and $\Pr[y^*(q)=-1]=1-\Pr[y^*(q)=1]$ obtained with the joint alignment model, and then compute the conditional probability
\begin{align}
\mathds{P}(Q^+\,|\,Q) =\prod_{q\in Q^+}\Pr\big[y^*(q)=1\big] \prod_{q\in Q\setminus Q^+} \Pr\big[y^*(q)=-1\big].
\end{align}

Furthermore, the expected overall inference power of $Q$ is
\begin{align}
    \mathds{E}_{\mathds{P}(Q^+\,|\,Q)}\big[\mathcal{I}(P\,|\,Q^+)\big] 
    = \mathds{E}_{\mathds{P}(Q^+\,|\,Q)}\Big[\sum_{q'\in P}\mathcal{I}(q'\,|\,Q^+)\Big],
\end{align}
where $P$ is the element pair pool.

Now, we want to select a set of element pairs that can maximize the expected overall inference power. 
Formally, we define the element pair selection problem as follows:
\begin{align}
\begin{split}
    \max \quad & \mathds{E}_{\mathds{P}(Q^+\,|\,Q)}\big[\mathcal{I}(P\,|\,Q^+)\big] , \\
    \text{s.t.} \quad  & |\,Q\,| = B \wedge Q\subseteq P,
\end{split}
\end{align}
where $B$ is the number constraint of selected element pairs.
Here, we can pick not only entity pairs, but also class and relation pairs.

Below, we present two approximation algorithms to solve the above problem. 
They also measure the inference power inside.

\smallskip
\noindent\textbf{A greedy algorithm.}
We describe a greedy algorithm to solve the element pair selection problem.
For convenience, the gain of a new element pair $q$, denoted by $\mathcal{G}(q\,|\,Q)$, is defined as 
\begin{align}
    \mathcal{G}(q\,|\,Q) = \mathds{E}_{\mathds{P}(Q^+\,|\,Q\cup\{q\})}\big[\mathcal{I}(P\,|\,Q^+)\big] - \mathds{E}_{\mathds{P}(Q^+\,|\,Q)}\big[\mathcal{I}(P\,|\,Q^+)\big],
\end{align}
where $Q$ is the element pair set already selected.
As depicted in Algorithm~\ref{algo:greedy}, the algorithm starts from an empty element pair set $Q=\emptyset$.
It iteratively selects a new element pair $q$, which maximizes $\mathcal{G}(q\,|\,Q)$, into $Q$ until the size of $Q$ reaches number constraint $B$.

\begin{theorem}
\label{thm:submodular}
    $\mathds{E}_{\mathds{P}(Q^+\,|\,Q)}\big[\mathcal{I}(P\,|\,Q^+)\big]$ is an increasing sub-modular function.
\end{theorem}
\begin{proof}
Please see Appendix~\ref{app:submodular}.
\end{proof}

Therefore, the element pair selection problem is sub-modular, and the greedy algorithm gives a $(1-\frac{1}{e})$-approximation guarantee~\cite{Greedy}. 
Recall that we only consider at most $\mu$-hop paths in inference, and Line 2 is a brute-force step to compute $\mathcal{I}(P\,|\,q)$ through enumerating all $\mu$-hop paths for all $q$, which costs $O(|P|^{\mu+1})$ time.
Furthermore, finding the best $q$ in Lines 3--5 takes $O(B |P|^2)$ time. 
The overall time complexity of Algorithm~\ref{algo:greedy} is $O\big(B |P|^2 + |P|^{\mu+1}\big)$.

\begin{algorithm}[t]
\KwIn{pool $P$, number constraint $B$}
\KwOut{element pair set $Q$}
$Q\leftarrow\emptyset$\;
\lFor{$q\in P$}{compute $\mathcal{I}(P\,|\,q)$ based on Eq.~(\ref{eq:over_inf})}
\For{$i\leftarrow 1, 2,\dots,B$}{
    compute $\mathcal{G}(q\,|\,Q)$ for each $q\in P \setminus Q$\;
    $Q\leftarrow Q\cup \big\{\argmax_{q}\mathcal{G}(q\,|\,Q)\big\}$\;
}
\Return $Q$\;
\caption{Greedy selection}
\label{algo:greedy}
\end{algorithm}

\smallskip
\noindent\textbf{A graph partitioning-based algorithm.}
The brute-force step in Algorithm~\ref{algo:greedy} is time-consuming, and we intend to reduce $|P|$ for improving efficiency.
Specifically, we build a partitioning graph by partitioning element pairs in the alignment graph into mutually exclusive groups, and estimate $\mathcal{I}()$ between any two groups.
Suppose that we partition the pool $P$ into $n$ groups $\{P_i\}_{i=1}^n$.
For each element pair $q$, we compute the estimated inference power $\hat{\mathcal{I}}(P\,|\,q)$ based on a path from $q$ through partitions in $\{P_i\}_{i=1}^n$, and the time cost is reduced to $O(|P|\times n^\mu)$.
Based on the estimated inference power, we run Algorithm~\ref{algo:greedy} to find a solution.

After partitioning element pairs into groups, some paths in the original alignment graph introducing self-loops in the partitioning graph, $\hat{\mathcal{I}}(P\,|\,q)$ is smaller than $\mathcal{I}(P\,|\,q)$.
To preserve the quality of the solution, we set a threshold $\rho$ such that the drop of inference power from each element pair $q$ to its neighbors $\mathcal{N}(e)$ is smaller than $1-\rho$, i.e., $ \hat{\mathcal{I}}\big(\mathcal{N}(e)\,|\,q\big) \geq\rho \mathcal{I}\big(\mathcal{N}(e)\,|\,q\big)$.
As we only consider at most $\mu$-hop paths, $\hat{\mathcal{I}}(P\,|\,q) \geq \rho^\mu \mathcal{I}(P\,|\,q)$.
Based on this bound, we design an algorithm giving a $\rho^\mu (1-\frac{1}{e})$-approximation guarantee.

\begin{theorem}
\label{thm:part}
Algorithm~\ref{algo:part} gives a $\rho^\mu(1-\frac{1}{e})$-approximation guarantee.
\end{theorem}
\begin{proof}
Please see Appendix~\ref{app:part}.
\end{proof}

\begin{algorithm}[t]
\KwIn{pool $P$, number constraint $B$, threshold $\rho$, edge set of alignment graph $T''$}
\KwOut{element pair set $Q$}
\lForEach{$\big(q,(r,r'),q'\big)\in T''$}{compute $\mathcal{I}(q'\,|\,q)$}
$P_1\leftarrow Q$; $n\leftarrow 1$; $flag\leftarrow \text{TRUE}$\;
\While{$flag$}{
    $flag\leftarrow \text{FALSE}$\;
    \For{$i=1,2,\dots,n$}{
        \ForEach{$q\in P_i$}{
            $\mathcal{I}_{\text{inner}}\leftarrow\sum_{q'\in P_i\cap \mathcal{N}(q)} \mathcal{I}(q'\,|\,q)$\;
            $\mathcal{I}_{\text{outer}}\leftarrow\sum_{q'\in (P\setminus P_i)\cap \mathcal{N}(q)} \mathcal{I}(q'\,|\,q)$\;
            $\rho'\leftarrow \min\{\rho',\frac{\mathcal{I}_{\text{outer}}}{\mathcal{I}_{\text{inner}} + \mathcal{I}_{\text{outer}}}$\}\;
        }
        \If{$\rho' < \rho$}{
            $(r,r')\leftarrow \argmax_{(r,r')}\ \smashoperator{\sum\limits_{(q,q')\in P_i\times P_i:(q,(r,r'),q')\in T''}}\ \mathcal{I}(q'\,|\,q)$\;
            $P_{n+1}\leftarrow \big\{q\in P_i\,|\,(q,(r,r'),q')\in T''\big\}$\;
            $P_i\leftarrow P_i\setminus P_{n+1}$;
            $n\leftarrow n + 1$;
            $flag\leftarrow \text{TRUE}$\;
            \textbf{break}\;
        }
    }
}
\lFor{$q\in P$}{compute $\hat{\mathcal{I}}(P\,|\,q)$}
$Q\leftarrow$ call Algorithm~\ref{algo:greedy} with $\hat{\mathcal{I}}(\cdot)$ \;
\Return $Q$\;
\caption{Graph partitioning-based selection}
\label{algo:part}
\end{algorithm}

The detailed algorithm is shown in Algorithm~\ref{algo:part}.
In Line 1, it computes $\mathcal{I}(q'\,|\,q)$ for all edges in the alignment graph.
In Line 2, it initializes the partitioning with a trivial partition $\{P_1\}$, where $P_1=P$.
In Lines 3--14, it iteratively splits the partition containing an element pair $q$ that does not satisfy $\hat{\mathcal{I}}\big(\mathcal{N}(e)\,|\,q\big) \geq\rho \mathcal{I}\big(\mathcal{N}(e)\,|\,q\big)$, until no partition needs to be split, i.e., $flag = \text{FALSE}$.
In Lines 5--6, it enumerates each partition $P_i$ and checks each element pair $q\in P_i$.
In Lines 7--8, it computes the inference power from $q$ to $P_i$, denoted by $\mathcal{I}_{\text{inner}}$, and the total inference power from $q$ to $Q\setminus P_i$, denoted by $\mathcal{I}_{\text{outer}}$.
Note that $1 - \frac{\mathcal{I}_{\text{outer}}}{\mathcal{I}_{\text{inner}} + \mathcal{I}_{\text{outer}}}$ is the drop of inference power, and Line 9 can find the minimal $\frac{\mathcal{I}_{\text{outer}}}{\mathcal{I}_{\text{inner}} + \mathcal{I}_{\text{outer}}}$ after iteration.
In Line 10, it checks whether the drop of inference power exceeds the threshold $\rho$.
If TRUE, Lines 11--14 find the relation pair $(r,r')$ that appears in most paths from $P_i$ to itself, and split $P_i$ based on $(r,r')$.
In Lines 15--16, it computes $\hat{\mathcal{I}}()$, and finds the element pair set $P_j$ with Algorithm~\ref{algo:greedy}.
As the main loop runs at most $|P|$ times, and $n \leq |P|$, Lines 3--14 cost $O(|P|^3)$ time. 
Therefore, the overall time complexity of Algorithm~\ref{algo:part} is $O\big(|P|^3 + B |P| n + |P| n^{\mu}\big)$. 
Since $|P|$ is usually much larger than $n$, when $\mu\geq 2$, Algorithm~\ref{algo:part} is more efficient than Algorithm~\ref{algo:greedy}.
\section{Experiments and Results}
\label{sect:exp}

In this section, we compare the proposed approach \modelname with the state-of-the-arts, and evaluate the effectiveness of its modules. 

\subsection{Experiment Preparation}

\noindent\textbf{Implementation.} 
We develop \modelname in Python 3.9 on a workstation with an Intel Xeon 3.3GHz CPU, 128GB memory and an NVIDIA GeForce 2080 Ti GPU card. 
We choose TransE \cite{TransE}, RotatE \cite{RotatE} and CompGCN \cite{CompGCN} as base entity-relation embedding models.
We set the embedding dimension of TransE and RotatE $dim = 100$ as suggested in~\cite{OpenEA}, and the dimension of CompGCN $dim = 200$ following~\cite{CompGCN}.
We search for the dimension of class embeddings in $\{50, 100, 150, 200\}$, and finally pick $50$.
For alignment probability calibration, we set $Z_{\text{ent}}=0.05$ and $Z_{\text{rel}}=Z_{\text{cls}}=0.1$, such that most true element matches have an alignment probability greater than $0.5$.
The similarity threshold is $\tau=0.9$ in semi-supervised learning.
We select $N=1,000$ nearest entities to build the candidate pool.
In active learning, the inference power threshold is $\kappa=0.8$, the partition threshold is $\rho=0.9$, and the maximum hops of inference paths is $\mu=5$.
At each time, active learning picks $B=100$ element pairs to probe the oracle.

\begin{table}[!t]
\caption{Dataset Statistics}
\label{tab:dataset}
\centering
{\small
\begin{tabular}{@{}c|cllc@{}}
\toprule Datasets & Entities & Relations & Classes & Matches \\
\midrule 
    D-W   & 100,000 vs. 70,000 & 413 vs. 261 & 167 vs. 116 & 70,193 \\
    D-Y   & 100,000 vs. 70,000 & 287 vs. 32  & \ \ 13 vs. 9 & 70,030 \\
    EN-DE & 100,000 vs. 70,000 & 381 vs. 196 & 109 vs. 76  & 70,248 \\
    EN-FR & 100,000 vs. 70,000 & 400 vs. 300 & 174 vs. 121 & 70,308 \\
\bottomrule
\end{tabular}}
\end{table}

\smallskip
\noindent\textbf{Datasets.} 
We use the benchmark datasets in OpenEA~\cite{OpenEA} to evaluate the performance of \modelname.
The datasets are constructed by sampling entities from three well-known real-world KGs: multi-lingual DBpedia \cite{DBpedia}, Wikidata \cite{Wikidata} and YAGO \cite{YAGO}.
There are four datasets, namely DBpedia-Wikidata (D-W), DBpedia-YAGO (D-Y), English DBpedia-German DBpedia (EN-DE), and English DBpedia-French DBpedia (EN-FR).
Each dataset contains two sets of triplets and a set of gold entity matches.
Entities in different language editions of DBpedia are extracted from the infobox, and their triplets are structured with the same schemata.
For the D-W dataset, we retrieve gold schema matches from the DBpedia SPARQL endpoint.
For the D-Y dataset, we retrieve gold schema matches from the YAGO website.
Following~\cite{ActiveEA}, we modify the datasets by removing $30\%$ of entities from the second KG to evaluate the robustness of deep active learning with dangling entities.
We list the statistics of these datasets in Table~\ref{tab:dataset}.

\smallskip
\noindent\textbf{Evaluation metrics.}
We use two sets of evaluation metrics.
The first set is H@$k$ ($k=1,10$) and Mean Reciprocal Rank (MRR), which are widely used in KG embedding.
H@$k$ is calculated by measuring the proportion of true matches within the top-$k$ nearest neighbors of each element. 
H@1 is also the accuracy of element alignment.
MRR is the average value of the reciprocal ranks of true matches.
The second set is precision, recall and F1-score, which are broadly used by conventional methods.
We follow \cite{F1EA} to compute F1-score with a greedy matching strategy.
We repeat the experiments five times, and report the average scores and run-time.
Due to the space limitation, we present H@1, MRR, F1-score and run-time in this section, and supplement H@10, precision and recall online.

\subsection{Overall Performance}

We compare \modelname with both deep KG alignment methods and active KG alignment algorithms.
Our hypotheses are to verify that our joint alignment model can improve the performance of both entity and schema alignment, and our active learning algorithm can achieve better accuracy with the same number of labels.

\begin{table}[!t]
\caption{Performance comparison of deep alignment methods}
\label{tab:kga}
\resizebox{\textwidth}{!}{
\begin{tabular}{@{}l|lcccccccccccc@{}}
\toprule
  \multicolumn{2}{c}{} & \multicolumn{3}{c}{D-W} & \multicolumn{3}{c}{D-Y} & \multicolumn{3}{c}{EN-DE} & \multicolumn{3}{c}{EN-FR} \\ \cmidrule(lr){3-5}\cmidrule(lr){6-8}\cmidrule(lr){9-11}\cmidrule(l){12-14}
  \multicolumn{2}{c}{} & H@1 & MRR & F1 & H@1 & MRR & F1 & H@1 & MRR & F1 & H@1 & MRR & F1 \\ 
\midrule
  \parbox[c]{3mm}{\multirow{11}{*}{\rotatebox[origin=c]{90}{Entity alignment}}}
  & PARIS & 0.451 & - & 0.487 & 0.501 & - & 0.550 & 0.583 & - & 0.602 & 0.346 & - & 0.381 \\
  & MTransE & 0.209 & 0.350 & 0.263  & 0.395 & 0.490 & 0.465  & 0.347 & 0.431 & 0.363  & 0.234 & 0.336 & 0.264  \\
  & BootEA & 0.312 & 0.409 & 0.335  & 0.475 & 0.556 & 0.512  & 0.610 & 0.672 & 0.630  & 0.228 & 0.338 & 0.293  \\
  & GCN-Align & 0.307 & 0.410 & 0.366  & 0.290 & 0.387 & 0.348  & 0.436 & 0.523 & 0.460  & 0.250 & 0.362 & 0.310  \\
  & AttrE & 0.085 & 0.153 & 0.071  & 0.570 & 0.649 & 0.551  & 0.307 & 0.400 & 0.290  & 0.246 & 0.336 & 0.225  \\
  & RSN & 0.441 & 0.521 & 0.464  & 0.514 & 0.580 & 0.541  & 0.587 & 0.662 & 0.393  & 0.393 & 0.487 & 0.421  \\
  & MuGNN & 0.356 & 0.446 & 0.411  & 0.354 & 0.428 & 0.489  & 0.580 & 0.610 & 0.657  & 0.292 & 0.418 & 0.376  \\
  & MultiKE & 0.236 & 0.255 & 0.315  & 0.718 & 0.726 & 0.736   & 0.576 & 0.604 & 0.609  & \textbf{0.608} & 0.638 & 0.626  \\
  & KECG & 0.632 & 0.726 & 0.692  & 0.728 & 0.795 & 0.765  & 0.625 & 0.711 & 0.682  & 0.565 & 0.655 & 0.635  \\
  & \modelname (CompGCN) & \textbf{0.654} & \textbf{0.742} & \textbf{0.741}  & \textbf{0.757} & \textbf{0.814} & \textbf{0.847}  & \textbf{0.657} & \textbf{0.716} & \textbf{0.699}  & 0.584 & \textbf{0.669} & \textbf{0.658}  \\
\midrule
\parbox[c]{3mm}{\multirow{11}{*}{\rotatebox[origin=c]{90}{Relation alignment}}}
  & PARIS & 0.432 & -  & 0.310 & 0.476 & -  & 0.360 & 0.594 & -  & 0.348 & 0.484 & - & 0.376 \\
  & MTransE & 0.000 & 0.043 & 0.000 & 0.111 & 0.293 & 0.133 & 0.003 & 0.027 & 0.006 & 0.000 & 0.023 & 0.000 \\
  & BootEA & 0.026 & 0.074 & 0.039 & 0.056 & 0.288 & 0.113 & 0.000 & 0.064 & 0.000 & 0.005 & 0.027 & 0.011 \\
  & GCN-Align & 0.019 & 0.072 & 0.042 & 0.056 & 0.303 & 0.122 & 0.000 & 0.031 & 0.000 & 0.005 & 0.036 & 0.017 \\
  & AttrE & 0.031 & 0.040 & 0.048 & 0.053 & 0.059 & 0.057 & 0.157 & 0.163 & 0.183 & 0.184 & 0.285 & 0.198 \\
  & RSN & 0.047 & 0.057 & 0.052 & 0.053 & 0.060 & 0.083 & 0.060 & 0.069 & 0.071 & 0.043 & 0.051 & 0.049 \\
  & MuGNN & 0.030 & 0.118 & 0.078 & 0.098 & 0.524 & 0.131 & 0.009 & 0.057 & 0.013 & 0.010 & 0.069 & 0.044 \\
  & MultiKE & 0.308 & 0.407 & 0.400 & 0.306 & 0.438 & 0.367 & 0.393 & 0.471 & 0.398 & 0.418 & 0.485 & 0.451 \\
  & KECG & 0.064 & 0.077 & 0.072 & 0.074 & 0.081 & 0.116 & 0.066 & 0.073 & 0.109 & 0.059 & 0.066 & 0.106 \\
  & \modelname (CompGCN) & \textbf{0.689} & \textbf{0.767} & \textbf{0.766} & \textbf{0.775} & \textbf{0.840} & \textbf{0.828} & \textbf{0.682} & \textbf{0.736} & \textbf{0.776} & \textbf{0.597} & \textbf{0.715} & \textbf{0.661} \\
\midrule
\parbox[c]{3mm}{\multirow{11}{*}{\rotatebox[origin=c]{90}{Class alignment}}}
  & PARIS & 0.288 & - & 0.315 & 0.511 & - & 0.569 & 0.297 & - & 0.349 & 0.136 & - & 0.157\\ 
  & MTransE & 0.004 & 0.131 & 0.000 & 0.056 & 0.278 & 0.133 & 0.000 & 0.065 & 0.000 & 0.000 & 0.084 & 0.000 \\
  & BootEA & 0.026 & 0.186 & 0.043 & 0.111 & 0.306 & 0.180 & 0.000 & 0.025 & 0.000 & 0.000 & 0.091 & 0.000 \\
  & GCN-Align & 0.026 & 0.201 & 0.067 & 0.056 & 0.218 & 0.141 & 0.000 & 0.059 & 0.000 & 0.000 & 0.099 & 0.000 \\
  & AttrE & 0.058 & 0.122 & 0.069 & 0.153 & 0.274 & 0.267 & 0.041 & 0.048 & 0.046 & 0.097 & 0.165 & 0.153 \\
  & RSN & 0.047 & 0.052 & 0.089 & 0.055 & 0.261 & 0.115 & 0.060 & 0.068 & 0.064 & 0.044 & 0.049 & 0.063 \\
  & MuGNN & 0.035 & 0.263 & 0.054 & 0.080 & 0.294 & 0.119 & 0.000 & 0.085 & 0.000 & 0.061 & 0.142 & 0.103 \\
  & MultiKE & 0.181 & 0.249 & 0.170 & 0.222 & 0.300 & 0.247 & 0.214 & 0.292 & 0.268 & 0.282 & 0.367 & 0.349 \\
  & KECG & 0.065 & 0.073 & 0.068 & 0.078 & 0.284 & 0.136 & 0.067 & 0.075 & 0.085 & 0.061 & 0.068 & 0.086 \\
  & BERTMap & 0.629 & -  & 0.469  & 0.692 & -  & 0.813  & 0.289 & -  & 0.307  & 0.376 & -  & 0.384  \\
  & \modelname (CompGCN) & \textbf{0.661} & \textbf{0.760} & \textbf{0.746} & \textbf{0.785} & \textbf{0.843} & \textbf{0.853} & \textbf{0.667} & \textbf{0.717} & \textbf{0.724} & \textbf{0.633} & \textbf{0.672} & \textbf{0.750} \\
\bottomrule
\end{tabular}}
\end{table}

\begin{table}[!t]
\caption{Running time comparison of deep alignment methods}
\label{tab:runtime}
\resizebox{\textwidth}{!}{
\begin{tabular}{@{}lcccc|lcccc@{}}
\toprule
& D-W & D-Y & EN-DE & EN-FR & & D-W & D-Y & EN-DE & EN-FR \\
\midrule
PARIS     & 7.12s & 7.74s & 10.6s & 7.01s & DAAKG (TransE)              & 2.13h & 2.68h & 2.59h & 2.16h \\
MTransE   & 0.20h & 0.26h & 0.23h & 0.21h & \ \ w/o class embeddings & 1.68h & 1.89h & 2.04h & 1.67h \\
BootEA    & 1.98h & 2.59h & 2.32h & 2.14h & \ \ w/o mean embeddings  & 1.86h & 2.25h & 2.41h & 1.61h \\
GCN-Align & 1.07h & 1.34h & 1.17h & 1.10h & \ \ w/o semi-supervision & 0.23h & 0.33h & 0.32h & 0.27h \\
AttrE     & 2.10h & 2.13h & 2.80h & 3.04h & DAAKG (RotatE)               & 2.26h & 2.91h & 2.56h & 2.34h \\
RSN       & 3.18h & 3.98h & 3.61h & 3.45h & \ \ w/o class embeddings & 1.70h & 2.01h & 2.01h & 1.81h \\
MuGNN     & 1.71h & 2.38h & 2.09h & 1.87h & \ \ w/o mean embeddings  & 1.79h & 2.34h & 2.10h & 1.90h \\
MultiKE   & 0.67h & 0.80h & 0.79h & 1.32h & \ \ w/o semi-supervision & 0.25h & 0.35h & 0.30h & 0.29h \\
KECG      & 1.97h & 2.69h & 2.40h & 2.24h & DAAKG (CompGCN)                & 2.37h & 2.99h & 2.39h & 2.52h \\
BERTMap   & 142s  & 68.7s & 107s  & 130s  & \ \ w/o class embeddings & 1.82h & 2.20h & 1.85h & 1.86h \\
DAAKG (CompGCN)  & 2.37h & 2.99h & 2.39h & 2.52h & \ \ w/o mean embeddings  & 1.96h & 2.59h & 2.09h & 1.97h \\
\cline{1-5} & & & & & \ \ w/o semi-supervision & 0.27h & 0.38h & 0.29h & 0.30h \\
\cline{6-10}
\end{tabular}}
\end{table}

\smallskip
\noindent\textbf{Performance of deep alignment.} 
To our best knowledge, most of existing deep KG alignment methods focus on entity alignment.
However, some of them can conduct schema alignment by treating classes as entities. 
For deep KG alignment, we compare \modelname with nine competitors, including MTransE~\cite{MTransE}, BootEA~\cite{BootEA}, GCN-Align~\cite{GCN-Align}, AttrE~\cite{AttrE}, RSN~\cite{RSN}, MuGNN~\cite{MuGNN}, MultiKE~\cite{MultiKE}, KECG \cite{KECG} and BERTMap~\cite{BERTMap}.
Note that BERTMap only aligns classes, while AttrE and MultiKE additionally leverages literal attributes of entities.
See Sect.~\ref{subsect:deep} for details.
We also compare with PARIS \cite{PARIS}, a conventional probabilistic method for aligning instances, classes, and relations simultaneously across ontologies. 

Table~\ref{tab:kga} presents the comparison results, and we make several observations: 
(1) For entity alignment, \modelname outperforms all competitors in most cases.
On average, \modelname improves H@1 by 0.026, MRR by 0.014 and F1-score by 0.043, compared with the second-best method, KECG.
This is because \modelname can benefit from the joint embedding learning with schema alignment.
(2) For relation and class alignment, only \modelname can achieve satisfactory results, BERTMap and PARIS are average in performance, and all deep entity alignment competitors perform poorly.
The primary reason for this is that these entity alignment competitors are not intended for schema alignment.
They treat classes as entities in embedding learning, which causes the significant loss of ontology semantics.
Compared with BERTMap, \modelname improves H@1 by 0.378 on EN-DE and 0.258 on EN-FR.
The reason is that BERTMap fails to handle multilingual class names, while \modelname leverages KG structures to align classes.
On the other hand, BERTMap can efficiently align classes of monolingual names without the entity-class information.
The results of \modelname and PARIS also indicate that entity and schemata alignment mutually benefit.
(3) Among the deep competitors, MultiKE gains the best overall results, indicating that both entity and schema alignment benefit from literal attributes.
However, MultiKE is sensitive to literal attributes, e.g., it only obtains 0.236 and 0.315 of H@1 and F1-score on the D-W dataset, respectively, as entity names in DBpedia and Wikidata are quite different.
\modelname achieves comparable results with MultiKE, due to its complex GNN encoder and semi-supervision.
We plan to incorporate \modelname with literal encoders for further performance improvement.

From Table~\ref{tab:runtime}, we observe that PARIS runs pretty fast, because it needs no training. 
Compared to it, all deep methods cost more run-time. 
\modelname runs relatively slow, due to semi-supervision and schema alignment.
As BERTMap only processes classes, it spends several minutes.

\smallskip
\noindent\textbf{Performance of active alignment.} 
For active KG alignment, we choose the following competitors:
\emph{Random} is the default element pair selection algorithm used in the training set construction.
\emph{Degree} selects element pairs with the largest degrees in the alignment graph.
\emph{PageRank} picks element pairs with the highest PageRank scores.
\emph{Uncertainty} finds element pairs with the most uncertain prediction, i.e., the largest entropy of predicted alignment probability, which is widely employed in active entity alignment methods for tabular data~\cite{Corleone,DTAL}. 
\emph{ActiveEA}~\cite{ActiveEA} calculates the uncertainty of model prediction, and identifies element pairs which are most effective in reducing the overall uncertainty of its neighbors.

We assess the H@1 and F1, namely the progressive scores on the unseen test dataset of the joint alignment model with $10\%, 20\%, 30\%,$ $40\%$ and $50\%$ of training data selected by different active learning algorithms. 
For Uncertainty, ActiveEA and \modelname, we employ each of them to iteratively pick up to 100 element pairs to train the model.
It is worth noting that these datasets assume that each element matches at most one other element, and all deep alignment methods leverage this restriction to remove non-matches.
Following~\cite{ActiveEA}, we compute the proportion of labeled matches, instead of all element pairs in the pool.

Figure~\ref{fig:al} shows the results, and we make several observations:
(1) \modelname consistently achieves the best H@1 and F1-score on all datasets.
This demonstrates the effectiveness of \modelname.
The main reason is that the batch active learning in \modelname can find element pairs with the largest overall inference power, while the selected element pairs from the competitors may infer each other in the same batch.
(2) ActiveEA is the second-best algorithm, which shows that the structural sampling can also capture inference power.
(3) With the same number of labels, CompGCN works slightly better than RotatE and TransE, e.g., compared with RotatE, CompGCN improves H@1 by 0.054 and F1-score by 0.041 with 50\% of training data.
This is mainly related to the performance of the models themselves.

\begin{figure}
    \centering
    \includegraphics[width=\textwidth]{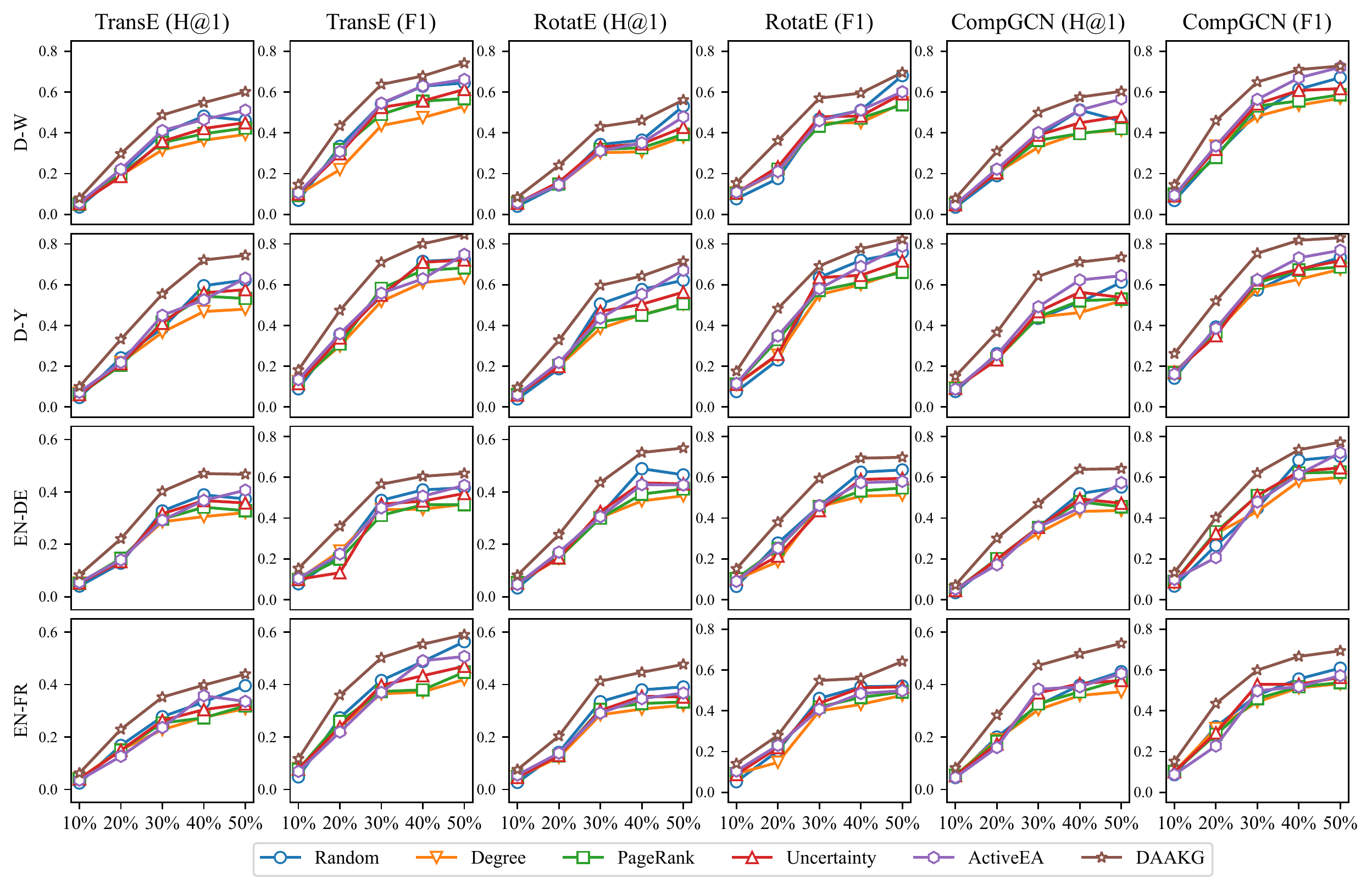}
    \caption{H@1 and F1 comparison of active alignment algorithms}
    \label{fig:al}
\end{figure}

\begin{table}[!ht]
\caption{Ablation study of embedding-based joint alignment}
\label{tab:ablation}
\resizebox{\textwidth}{!}{
\begin{tabular}{l|lccccccccccccc}
\toprule
  \multicolumn{2}{c}{\multirow{2}{*}{H@1}} & \multicolumn{3}{c}{D-W} & \multicolumn{3}{c}{D-Y} & \multicolumn{3}{c}{EN-DE} & \multicolumn{3}{c}{EN-FR} \\
  \cmidrule(lr){3-5}\cmidrule(lr){6-8}\cmidrule(lr){9-11}\cmidrule(l){12-14}
  \multicolumn{2}{c}{} & Ent. & Rel. & Cls. & Ent. & Rel. & Cls. & Ent. & Rel. & Cls. & Ent. & Rel. & Cls. \\ 
\midrule
\parbox[c]{3mm}{\multirow{4}{*}{\rotatebox[origin=c]{90}{TransE}}}
  & \modelname & 0.608 & 0.543 & 0.514  & 0.752 & 0.771 & 0.701  & 0.492 & 0.669 & 0.649  & 0.470 & 0.541 & 0.608  \\
  & \ \ w/o class embeddings & 0.542 & 0.480 & 0.477  & 0.751 & 0.759 & 0.619  & 0.465 & 0.604 & 0.616  & 0.441 & 0.469 & 0.503  \\
  & \ \ w/o mean embeddings & 0.547 & 0.188 & 0.257  & 0.701 & 0.367 & 0.429  & 0.415 & 0.295 & 0.226  & 0.456 & 0.361 & 0.328  \\
  & \ \ w/o semi-supervision & 0.525 & 0.388 & 0.390  & 0.745 & 0.551 & 0.463  & 0.357 & 0.401 & 0.526  & 0.189 & 0.468 & 0.390   \\
\midrule
\parbox[c]{3mm}{\multirow{4}{*}{\rotatebox[origin=c]{90}{RotatE}}}
  & \modelname & 0.566 & 0.588 & 0.516  & 0.731 & 0.703 & 0.681  & 0.580 & 0.659 & 0.541  & 0.496 & 0.487 & 0.599  \\
  & \ \ w/o class embeddings & 0.463 & 0.443 & 0.501  & 0.613 & 0.638 & 0.570  & 0.458 & 0.569 & 0.529  & 0.392 & 0.409 & 0.526  \\
  & \ \ w/o mean embeddings & 0.472 & 0.574 & 0.428  & 0.635 & 0.660 & 0.656  & 0.489 & 0.633 & 0.476  & 0.492 & 0.427 & 0.549  \\
  & \ \ w/o semi-supervision & 0.372 & 0.304 & 0.275  & 0.517 & 0.637 & 0.543  & 0.445 & 0.593 & 0.432  & 0.229 & 0.433 & 0.566  \\
\midrule
\parbox[c]{3mm}{\multirow{4}{*}{\rotatebox[origin=c]{90}{CompGCN}}}
  & \modelname & 0.654 & 0.689 & 0.661  & 0.757 & 0.775 & 0.785  & 0.657 & 0.682 & 0.667  & 0.584 & 0.597 & 0.633  \\
  & \ \ w/o class embeddings & 0.600 & 0.688 & 0.599  & 0.722 & 0.742 & 0.663  & 0.634 & 0.588 & 0.531  & 0.488 & 0.447 & 0.486  \\
  & \ \ w/o mean embeddings & 0.597 & 0.643 & 0.563  & 0.713 & 0.738 & 0.765  & 0.573 & 0.602 & 0.664  & 0.496 & 0.578 & 0.605  \\
  & \ \ w/o semi-supervision & 0.536 & 0.684 & 0.568  & 0.727 & 0.714 & 0.537  & 0.357 & 0.515 & 0.397  & 0.426 & 0.553 & 0.625  \\
\bottomrule
\toprule
  \multicolumn{2}{c}{\multirow{2}{*}{F1}} & \multicolumn{3}{c}{D-W} & \multicolumn{3}{c}{D-Y} & \multicolumn{3}{c}{EN-DE} & \multicolumn{3}{c}{EN-FR} \\ 
  \cmidrule(lr){3-5}\cmidrule(lr){6-8}\cmidrule(lr){9-11}\cmidrule(l){12-14}
  \multicolumn{2}{c}{} & Ent. & Rel. & Cls. & Ent. & Rel. & Cls. & Ent. & Rel. & Cls. & Ent. & Rel. & Cls. \\ 
\midrule
\parbox[c]{3mm}{\multirow{4}{*}{\rotatebox[origin=c]{90}{TransE}}}
  & \modelname & 0.692 & 0.677 & 0.637 & 0.852 & 0.850 & 0.807 & 0.621 & 0.710 & 0.701 & 0.635 & 0.690 & 0.722 \\
  & \ \ w/o class embeddings & 0.656 & 0.603 & 0.601 & 0.836 & 0.840 & 0.665 & 0.564 & 0.683 & 0.672 & 0.557 & 0.579 & 0.618 \\
  & \ \ w/o mean embeddings & 0.640 & 0.252 & 0.328 & 0.773 & 0.412 & 0.534 & 0.562 & 0.291 & 0.281 & 0.517 & 0.313 & 0.384 \\
  & \ \ w/o semi-supervision & 0.642 & 0.518 & 0.520 & 0.823 & 0.659 & 0.561 & 0.443 & 0.570 & 0.657 & 0.485 & 0.564 & 0.528 \\
\midrule
\parbox[c]{3mm}{\multirow{4}{*}{\rotatebox[origin=c]{90}{RotatE}}}
  & \modelname & 0.675 & 0.692 & 0.634 & 0.797 & 0.765 & 0.742 & 0.686 & 0.745 & 0.655 & 0.617 & 0.609 & 0.701 \\
  & \ \ w/o class embeddings & 0.588 & 0.570 & 0.622 & 0.635 & 0.729 & 0.657 & 0.584 & 0.677 & 0.645 & 0.522 & 0.538 & 0.643 \\
  & \ \ w/o mean embeddings & 0.596 & 0.681 & 0.556 & 0.710 & 0.697 & 0.713 & 0.611 & 0.726 & 0.600 & 0.614 & 0.555 & 0.661 \\
  & \ \ w/o semi-supervision & 0.502 & 0.430 & 0.397 & 0.577 & 0.701 & 0.676 & 0.572 & 0.696 & 0.560 & 0.342 & 0.561 & 0.675 \\
\midrule
\parbox[c]{3mm}{\multirow{4}{*}{\rotatebox[origin=c]{90}{CompGCN}}}
  & \modelname & 0.741 & 0.766 & 0.746 & 0.847 & 0.828 & 0.853 & 0.699 & 0.776 & 0.724 & 0.658 & 0.661 & 0.750 \\
  & \ \ w/o class embeddings & 0.701 & 0.765 & 0.701 & 0.784 & 0.796 & 0.687 & 0.679 & 0.641 & 0.634 & 0.579 & 0.582 & 0.558 \\
  & \ \ w/o mean embeddings & 0.699 & 0.733 & 0.673 & 0.774 & 0.809 & 0.824 & 0.591 & 0.637 & 0.692 & 0.610 & 0.722 & 0.727 \\
  & \ \ w/o semi-supervision & 0.651 & 0.762 & 0.677 & 0.802 & 0.761 & 0.609 & 0.515 & 0.621 & 0.523 & 0.505 & 0.658 & 0.692 \\
\bottomrule
\end{tabular}}
\end{table}

\subsection{Further Analysis}

\noindent\textbf{Analysis of embedding-based joint alignment.}
We consider three KG embedding models, i.e., TransE, RotatE and CompGCN, to encode the entity-relation structures in \modelname. 
Additionally, we compare \modelname with three variants: 
\modelname w/o class embeddings, which treats classes as entities to embed. 
\modelname w/o mean embeddings, which aligns classes and relations only with their original embeddings.
\modelname w/o semi-supervision, which does not leverage semi-supervised learning.

Table~\ref{tab:ablation} depicts the results. 
We can observe that (1) compared to TransE and RotatE, \modelname with CompGCN consistently achieves the best H@1 and F1-score on all datasets.
The major reason is that CompGCN can aggregate the neighboring information for entities, which enables the alignment model to make more accurate comparison based on neighboring entities.
(2) Class embeddings improve class alignment by 0.078 of H@1 and 0.081 of F1-score on average, which shows that the dedicated entity-class scoring function generates better class embeddings.
(3) Mean embeddings are the most important component for schema alignment.
The reason is that schema matches are far fewer than entity matches, and the alignment model underfits to schema alignment. 
The mean embeddings leverage entity matches to improve the quality of schema embeddings.
(4) Both class and mean embeddings improve H@1 and F1-score on entity alignment.
The reason may be that, although they have no direct impact on entity alignment, the class and relation matches can affect the discriminability of entity embeddings.
(5) As shown in Table~\ref{tab:runtime}, semi-supervision is the most time-consuming component. 
However, it improves H@1 by 0.145 and F1-score by 0.126 on average, indicating that potential matches help the alignment model discover more matches.

\begin{table}[!t]
\caption{Accuracy of inference power measurement}
\label{tab:inf}
{\small
  \begin{tabular}{@{}lcccc@{}}
  \toprule
  & D-W & D-Y & EN-DE & EN-FR \\
  \midrule
  TransE & 0.933 & 0.948 & 0.977 & 0.965 \\
  RotatE & 0.844 & 0.824 & 0.953 & 0.957 \\
  CompGCN & 0.772 & 0.763 & 0.872 & 0.864 \\
  \bottomrule
\end{tabular}}
\end{table}

\smallskip
\noindent\textbf{Analysis of inference power measurement.}
To evaluate the effectiveness of inference power measurement, we measure the proportion of element matches in the inferred element pairs, i.e., the accuracy of inference power measurement.
Table~\ref{tab:inf} reports the results with the three entity-relation embedding models, and we make two main observations:
(1) The lowest accuracy of inference power is 0.772, reflecting that our inference power measurement only makes 22.8\% of errors.
This good accuracy helps the selection algorithm effectively compare element pairs.
Note that, as unlabeled elements are aligned by the joint alignment model trained with additional new labels, it would not decrease the accuracy of model prediction.
(2) TransE has the most accurate inference power measurement.
RotatE comes in second, and CompGCN is the worst.
The reason is that the bound of tail entities is more accurate for TransE compared with RotatE and CompGCN.

\begin{figure}
    \centering
    \includegraphics[width=.8\columnwidth]{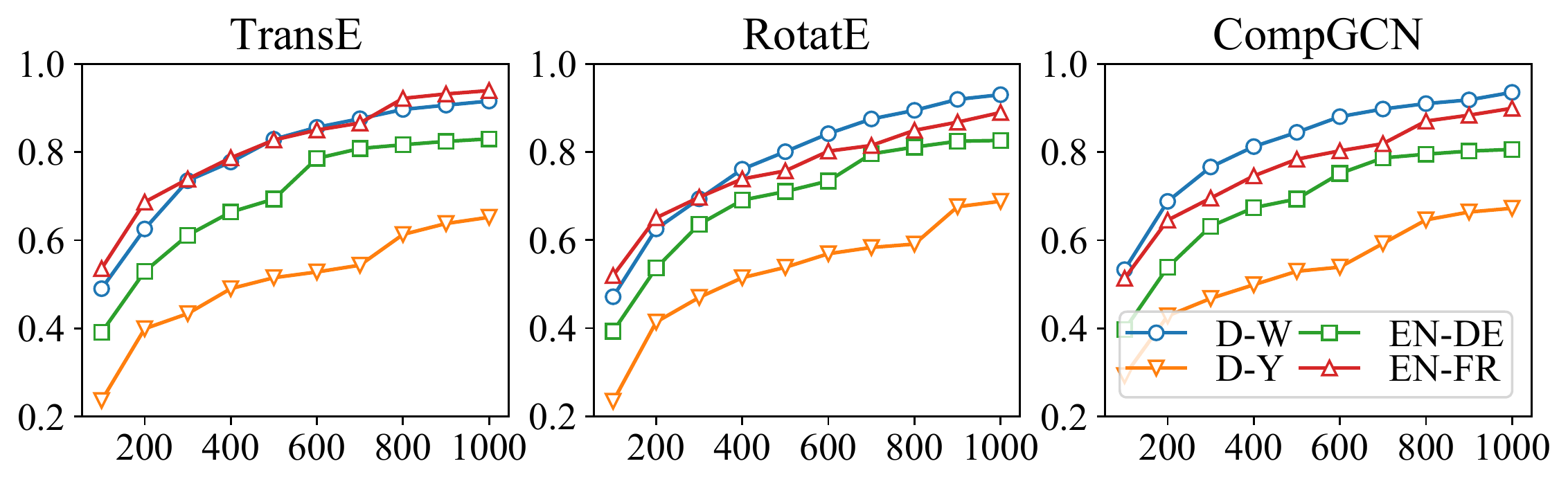}
    \caption{Recall of element pairs w.r.t. different $N$}
    \label{fig:blocking}
\end{figure}

\begin{figure}
    \centering
    \includegraphics[width=.75\columnwidth]{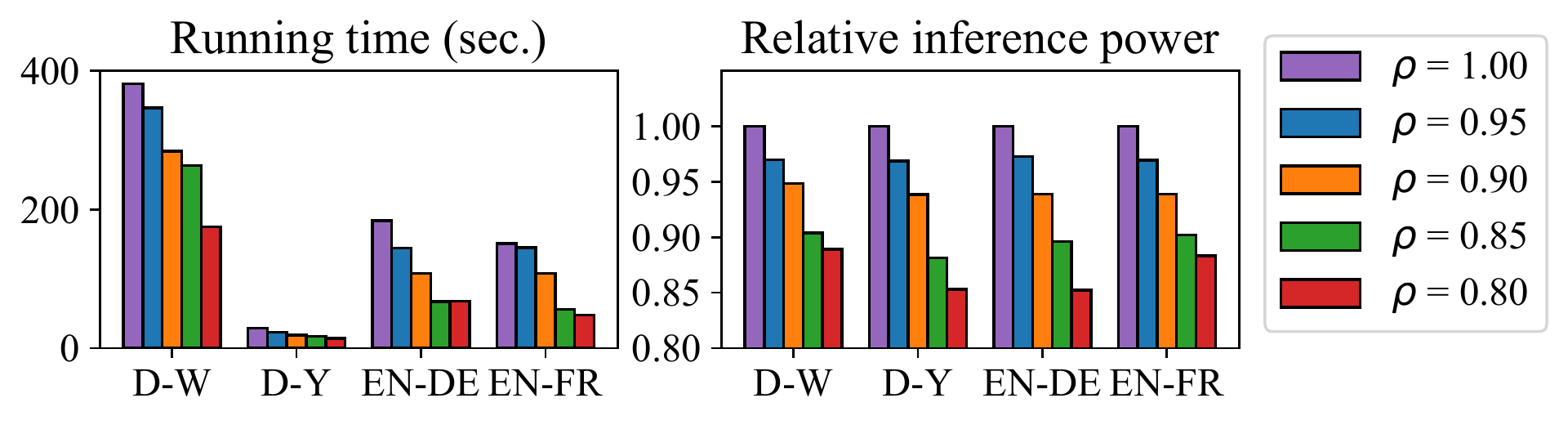}
    \caption{Running time and relative inference power of graph partitioning-based selection w.r.t. different $\rho$}
    \label{fig:batch}
\end{figure}

\smallskip
\noindent\textbf{Analysis of batch active learning.}
To evaluate the effectiveness of element pair pool generation, we set $N=100, 200,300,\dots,1000$ and compute the recall of reserved entity pairs in the pool.
Figure~\ref{fig:blocking} shows the results. 
We can find that (1) on the D-W, EN-DE and EN-FR datasets, pool generation can discover at least $80.6\%$ of element matches with top-$1000$ nearest neighbors, which reduces about $97.5\%$ of total element pairs.
This shows that pool generation can dramatically reduce the pool size while preserving a large portion of element matches.
(2) On the D-Y dataset, the recall ranges from $65.2\%$ to $68.8\%$, which is significantly worse than the other three datasets.
The main reason is that there are only 30 schema matches on the D-Y dataset, and most entities are described with similar sets of relations or classes, making the schema signatures less discriminating.

To evaluate the efficiency and effectiveness of graph partitioning-based selection, we set $\rho=1.00, 0.95, 0.90, 0.85, 0.80$ (when $\rho=1.00$, it degenerates to greedy selection), compute the relative inference power compared with greedy selection, and record the run-time.
As shown in Figure~\ref{fig:batch}, when $\rho$ decreases, the partitioning-based selection algorithm runs faster at the cost of inference power.
Extremely, when $\rho=0.80$, it is 2.5 times faster than greedy selection on average, while preserving at least $88\%$ of inference power.
\section{Conclusion}
\label{sect:concl}
In this paper, we propose a deep active alignment approach for jointly aligning entities and schemata in KGs.
We introduce a deep alignment model which encodes KG structures with both entity-relation and entity-class information, and aligns KG elements with different similarity functions.
As training deep alignment models requires a large number of labels, we introduce an active alignment approach.
We measure the inference power between element pairs based on the KG embedding model.
For batch active learning, we define an element pair selection problem that maximizes the expected overall inference power.
We prove that it is a sub-modular optimization problem that can be solved by a greedy algorithm.
To improve the efficiency, we further propose a graph partitioning-based algorithm with an approximation ratio guarantee.
The empirical evaluation shows that our approach significantly improves the accuracy of KG alignment under the same labeling budget.
In future work, we plan to extend our deep alignment model to leverage the side information (e.g., literals) of entities.
We also want to study active alignment for more types of KG alignment models.

\begin{acks}
This work is supported by the National Natural Science Foundation of China (No. 62272219) and the Alibaba Group through Alibaba Research Fellowship Program.
\end{acks}
\appendix

\section{Proof of Theorem~\ref{thm:submodular}}
\label{app:submodular}

The increment of inference power to one element pair $q'$ is
\begin{align*}
    \mathcal{I}\big(q'\,|\,Q^+\cup\{q\}\big) - \mathcal{I}(q'\,|\,Q^+) 
    = \big|\,\mathcal{I}(q'\,|\,q) - \mathcal{I}(q'\,|\,Q^+)\,\big|_+.
\end{align*}

We simplify the gain of a new element pair $q$ ($q\not\in Q$) as follows:
\begin{align*}
    \mathcal{G}(q\,|\,Q) & =\mathds{E}_{\mathds{P}(Q^+\,|\,Q\cup\{q\})}\big[\mathcal{I}(q'\,|\,Q^+)\big] - \mathds{E}_{\mathds{P}(Q^+\,|\,Q)}\big[\mathcal{I}(q'\,|\,Q^+)\big] \\
    & = \Pr\big[y^*(q)=0\big] \adjustlimits\sum_{q'\in P}\sum_{Q^+\subseteq Q} \mathds{P}(Q^+\,|\,Q)\,\mathcal{I}(q'\,|\,Q^+) \\
    &\quad + \Pr\big[y^*(q)=1\big] \adjustlimits\sum_{q'\in P}\sum_{Q^+\subseteq Q}\mathds{P}(Q^+\,|\,Q)\, \mathcal{I}\big(q'\,|\,Q^+\cup\{q\}\big) - \mathds{E}_{\mathds{P}(Q^+\,|\,Q)}\big[\mathcal{I}(q'\,|\,Q^+)\big] \\
    &= \Pr\big[y^*(q)=1\big] \adjustlimits\sum_{q'\in P}\sum_{Q^+\subseteq Q}\mathds{P}(Q^+\,|\,Q) \cdot \big|\mathcal{I}(q'\,|\,q) - \mathcal{I}(q'\,|\,Q^+)\big|_+.
\end{align*}

Obviously, $\mathcal{G}(q\,|\,Q) \geq 0$. 
Thus, $\mathds{E}_{\mathds{P}(Q^+\,|\,Q)}\big[\mathcal{I}(P\,|\,Q^+)\big]$ is an increasing function.

Let $Q'$ be another element pair set such that $Q\subseteq Q'$ and $q\not\in Q'$.
We compare $\mathcal{G}(q\,|\,Q)$ and $\mathcal{G}(q\,|\,Q')$.
Note that,
\begin{align*}
\mathcal{G}(q\,|\,Q') 
&= \Pr\big[y^*(q)=1\big] \adjustlimits\sum_{q'\in P}\sum_{Q^+\subseteq Q'}\mathds{P}(Q^+\,|\,Q') \cdot \big|\mathcal{I}(q'\,|\,q) - \mathcal{I}(q'\,|\,Q^+)\big|_+ \\
&\leq \Pr\big[y^*(q)=1\big] \adjustlimits\sum_{q'\in P}\sum_{Q_1^+\subseteq Q} \mathds{P}(Q_1^+\,|\,Q) \Big(\sum_{Q_2^+\subseteq Q'\setminus Q} \mathds{P}(Q_2^+\,|\,Q'\setminus Q) \\
&\qquad\qquad\qquad\qquad\qquad\qquad\qquad\qquad\qquad \times \big|\mathcal{I}(q'\,|\,q) - \mathcal{I}\big(q'\,|\,(Q_1^+\cup Q_2^+)\cap Q\big)\big|_+\Big).
\end{align*}

As $(Q_1^+\cup Q_2^+)\cap Q = Q_1^+$, and $\sum_{Q_2^+\subseteq Q'\setminus Q} \mathds{P}(Q_2^+\,|\,Q'\setminus Q)=1$, we have $\mathcal{G}(q\,|\,Q') \leq \mathcal{G}(q\,|\,Q)$, proving that $\mathds{E}_{\mathds{P}(Q^+\,|\,Q)}\big[\mathcal{I}(P\,|\,Q^+)\big]$ is sub-modular.

\section{Proof of Theorem~\ref{thm:part}}
\label{app:part}

We use $\mathcal{F}()$ to denote the optimization goal computed with the original inference power function, and $\hat{\mathcal{F}}()$ to denote the optimization goal computed with the estimated inference power function.
Due to $\rho^\mu \mathcal{I}(q'\,|\,q)\leq \hat{\mathcal{I}}(q'\,|\,q)\leq \mathcal{I}(q'\,|\,q)$, we have $\forall Q, \rho^\mu \mathcal{F}(Q)\leq\hat{\mathcal{F}}(Q)\leq \mathcal{F}(Q)$.

Let $Q^*$ be the optimal solution on the original alignment graph, $\hat{Q}^*$ be the optimal solution on the partitioning graph, and $Q$ be the solution obtained by Algorithm~\ref{algo:part}.
We have
\begin{align*}
    \mathcal{F}(Q) & \geq \hat{\mathcal{F}}(Q) \geq \Big(1-\frac{1}{e}\Big)\hat{\mathcal{F}}(\hat{Q}^*)\geq \Big(1-\frac{1}{e}\Big)\hat{\mathcal{F}}(Q^*) \geq \rho^\mu (1-\frac{1}{e})\mathcal{F}(Q^*).
\end{align*}

Thus, Algorithm~\ref{algo:part} gives a $\rho^\mu (1-\frac{1}{e})$-approximation guarantee.


\bibliographystyle{ACM-Reference-Format}
\bibliography{main}

\end{document}